\documentclass{article}
\usepackage{ijcai22}

\usepackage{times}
\usepackage{nicefrac}
\usepackage{soul}
\usepackage{natbib}
\usepackage{url}
\usepackage{hyperref}
\usepackage[utf8]{inputenc}
\usepackage[small]{caption}
\usepackage{graphicx}
\usepackage{graphbox}
\usepackage{bm}
\usepackage{amsmath}
\usepackage{amsfonts}
\usepackage{amsthm}
\usepackage{amssymb}
\usepackage{mathtools}
\usepackage{mathabx}
\usepackage{booktabs}
\usepackage{algorithm}
\usepackage{algorithmic}
\usepackage{paralist}
\usepackage{enumitem}
\usepackage{tikz}
\usetikzlibrary{decorations.pathreplacing,calligraphy}
\newtheorem{example}{Example}
\newtheorem{theorem}{Theorem}
\newtheorem{definition}{Definition}
\newtheorem{corollary}{Corollary}
\newtheorem{lemma}{Lemma}
\newtheorem{proposition}{Proposition}

 \newcommand{\cB}{\mathcal{B}} 
 \newcommand{\cD}{\mathcal{D}}

\newcommand{\cI}{\mathcal{I}} \newcommand{\cL}{\mathcal{L}}
\newcommand{\cM}{\mathcal{M}} \newcommand{\cN}{\mathcal{N}}
 
\newcommand{\cR}{\mathcal{R}} 
 
 \newcommand{\cJ}{\mathcal{J}}
 
 \newcommand{\cX}{\mathcal{X}}

\newcommand{\RR}{\mathbb{R}}

\newcommand{\noise}{\bm{\eta}}
\newcommand{\cop}{B^+}
\newcommand{\truedata}{\bm{x}}

\newcommand{\truealloc}{\alloc{\truedata}}
\newcommand{\noisyalloc}{\alloc{\noisydata}}
\newcommand{\noisydata}{\tilde{\bm{x}}}
\newcommand{\noisydataelem}{\tilde{x}}
\newcommand{\postdata}{\pp}

\newcommand{\nnpostdata}{\pp_+}
\newcommand{\sumpostdata}{\pp_S}
\newcommand{\nnsumpostdata}{\pp_{S+}}
\newcommand{\region}{\mathcal{K}}
\newcommand{\nnregion}{\mathcal{K}_+}
\newcommand{\sumregion}{\mathcal{K}_S}
\newcommand{\nnsumregion}{\mathcal{K}_{S+}}
\newcommand{\al}{P^F}
\newcommand{\pp}{\pi}

\newcommand{\bl}{\pp_{\mathrm{BL}}}
\newcommand{\pos}{\pp_{\mathrm{PoS}}}

\newcommand{\bias}[2]{\operatorname{\cB}\left(#1,#2\right)}
\newcommand{\biasm}[1]{\operatorname{\cB}\left(#1\right)}
\newcommand{\norm}[1]{\left\lVert#1\right\rVert}

\newcommand{\pr}[1]{\operatorname{Pr}\left(#1\right)}

\newcommand{\EE}[2]{\operatorname{\mathbb{E}}_{#1}\left[#2\right]}
\newcommand{\inner}[2]{\left\langle#1,#2\right\rangle}

\newcommand{\relu}[1]{\left(#1\right)_{\geq 0}}
\newcommand{\negpart}[1]{\left(#1\right)_-}
\newcommand{\alloc}[1]{\al\left(#1\right)}
\newcommand{\blmech}[1]{\bl\left(#1\right)}
\newcommand{\posmech}[1]{\pos\left(#1\right)}
\newcommand{\talloci}[1]{\al_{#1}\left(\truedata\right)}
\newcommand{\nalloci}[1]{\al_{#1}\left(\noisydata\right)}
\newcommand{\swap}[3]{\operatorname{Swap}_{#1,#2}\left(#3\right)}

\usepackage{todonotes}
\usepackage{setspace}

\title{Post-processing of Differentially Private Data: A Fairness Perspective}

\author{
Keyu Zhu$^1$
\and
Ferdinando Fioretto$^2$\And
Pascal Van Hentenryck$^{1}$
\affiliations
$^1$Georgia Institute of Technology\\
$^2$Syracuse University
\emails
kzhu67@gatech.edu,\,
ffiorett@syr.edu,\,
pvh@isye.gatech.edu
}

\begin{document}

\maketitle

\begin{abstract}
Post-processing immunity is a fundamental property of differential privacy: it enables arbitrary data-independent transformations to differentially private outputs without affecting their privacy guarantees. Post-processing is routinely applied in data-release applications, including census data, which are then used to make allocations with substantial societal impacts. This paper shows that post-processing causes disparate impacts on individuals or groups and analyzes two critical settings: the release of differentially private datasets and the use of such private datasets for downstream decisions, such as the allocation of funds  informed by US Census data. In the first setting, the paper proposes tight bounds on the unfairness for traditional post-processing mechanisms, giving a unique tool to decision makers to quantify the disparate impacts introduced by their release. In the second setting, this paper proposes a novel post-processing mechanism that is (approximately) optimal under different fairness metrics, either reducing fairness issues substantially or reducing the cost of privacy. The theoretical analysis is complemented with numerical simulations on Census data.
\end{abstract}

%%%%%%%%%%%%%%%%%%%%%%%%%%%%%%%%%%%%%%%%%%%%%%%%%%%%%%%%%%%%%%%%%%%%%%%%%%%%%%%%%%%%%%%%%%%%%%
\section{Introduction}
%%%%%%%%%%%%%%%%%%%%%%%%%%%%%%%%%%%%%%%%%%%%%%%%%%%%%%%%%%%%%%%%%%%%%%%%%%%%%%%%%%%%%%%%%%%%%%

Differential privacy (DP) \citep{Dwork:06} has become a fundamental technology for private data release. Private companies and federal agencies are rapidly developing their own implementations of DP. It is particularly significant to note that the U.S.~Census Bureau adopted DP for its 2020 release \citep{abowd2018us}. It is also of primary importance to observe that the released data by corporation or federal agencies are often used to make policy decisions with significant societal and economic impacts for the involved individuals. For example, U.S.~census data users rely on the decennial 
census data to apportion the 435 congressional seats, allocate the 
\$1.5 trillion budget, and distribute critical resources to U.S.~states and jurisdictions. 

Although DP provides strong privacy guarantees on the released data 
and is widely celebrated among privacy researchers, its wide adoption 
among more federal agencies and public policy makers presents a key 
challenge: without careful considerations, DP methods may 
disproportionately impact minorities in decision processes based on the private data. Specifically, to protect individuals in a dataset, typical DP data-release methods operate by adding calibrated noise onto the data and then \emph{post-process} the resulting noisy data to restore some important data invariants.  Since such a process perturbs the original data, it necessarily introduces some errors which propagate onto downstream decision tasks. In fact, this paper will show that these errors may affect various individuals differently. Although understanding the outcome of these effects is extremely important, these disproportionate impacts are poorly understood and have not received the attention they deserve given their broad impact on various population segments. 

{\em This paper addresses this gap in understanding the effect of DP, and analyzes the disproportionate effects of a family of post-processing methods commonly adopted in data release tasks.} 
% The analysis first focuses on the release of differentially private 
% datasets, it gives a tight bound on the unfairness resulting from commonly adopted post-processing steps, and shows that the presence of disparate impacts is unavoidable in these settings. 
% The paper then analyzes the disparate impacts in downstream decision processes that take as input private data, as those used by the U.S.~Census to allocate funds  and  benefits. Motivated by the negative result obtained for data-release, the paper proposes a novel post-processing mechanism which is applied on the decision problem itself, and shows that such mechanism is (approximately) optimal under different fairness metrics. 
The analysis focuses on two critical settings: the release of differentially private datasets and the use  of  such  private  datasets  in  critical  allocation tasks, as those using  U.S.~Census data to allocate funds  and  benefits. The paper makes two fundamental contributions:
\begin{enumerate}[leftmargin=*, parsep=0pt, itemsep=2pt, topsep=-4pt]
    \item In the release setting, the paper derives tight bounds on the unfairness introduced by commonly adopted post-processing mechanisms, providing a valuable tool for policy makers and information officers to understand the disproportionate impact of their DP releases. These results are complemented by numerical simulations on the census data.
    \item In the downstream decision setting, the paper proposes a novel post-processing mechanism that integrates the data invariant into the downstream decision processes. The resulting mechanism achieves near-optimal results and reduces unfairness and the cost of privacy up to an order of magnitude on practical case studies. 
\end{enumerate}
To the best of the authors' knowledge, this is the first study that analyzes the fairness impacts of DP post-processing steps. The rest of this paper presents the related work, the preliminaries, the settings considered, and the motivation. The core of the paper is in Sections
\ref{sec:census} and \ref{sec:alloc} that present the two main contributions. The last section concludes the paper. All the proofs are in the Appendices that also contain a nomenclature summary.

%%%%%%%%%%%%%%%%%%%%%%%%%%%%%%%%%%%%%%%%%%%%%%%%%%%%%%%%%%%%%%%%%%%%%%%%%%%%%%%%%%%%%%%%%%%%%%
\section{Related Work}
%%%%%%%%%%%%%%%%%%%%%%%%%%%%%%%%%%%%%%%%%%%%%%%%%%%%%%%%%%%%%%%%%%%%%%%%%%%%%%%%%%%%%%%%%%%%%%

Privacy and fairness have been studied mostly in isolation
with a few exceptions. 
\citet{cummings:19} considered the
tradeoffs arising between differential privacy and equal opportunity.
% a fairness concept that requires a classifier to produce equal true
% positive rates across different groups. They show that there exists no
% classifier that simultaneously achieves $(\epsilon,0)$-differential
% privacy, satisfies equal opportunity, and has accuracy better than a
% constant classifier.  
\citet{ekstrand:18} raised questions about the tradeoffs 
involved between privacy and fairness, and \citet{jagielski:18} 
showed two algorithms that satisfy $(\epsilon,\delta)$-differential privacy 
and equalized odds.
% Although it may sound like these algorithms contradict the
% impossibility result from \citep{cummings:19}, it is important to note
% that they are not considering an $(\epsilon, 0)$-differential privacy
% setting. Tran et al.~\citep{Tran:IJCAI21} developed a differentially private
% learning approach to enforce several group fairness notions using a Lagrangian dual method. 
In the context of data release and resource allocation, 
\citet{pujol2020fair} were seemingly first to show, empirically, 
that there might be privacy-fairness tradeoffs involved in resource 
allocation settings. In particular, for census data, 
they show that the noise added to achieve differential privacy could 
disproportionately affect some groups over others. 
\citet{DBLP:conf/ijcai/0007FHY21} formalized the ideas developed in 
\citep{pujol2020fair} and characterized the conditions for which
fairness violations can be bounded for a class of allocation 
problems. 
% \citep{zhu2021bias} studied projections provided a novel upper bound of the bias
% associated with projection (to be continued)
Finally, \citet{abowd2019economic} considered statistical accuracy 
and privacy protection as competing public goods, and designed an 
economic framework to balance the tradeoff. {\em This paper departs from these results significantly: it provides tight lower and upper bounds on the unfairness
introduced by post-processing steps that are critical for practical applications, and proposes new mechanisms that merge post-processing and the downstream resource allocation for mitigating these fairness issues.}

% \citep{pujol2020fair} demonstrated, with 3 practical case studies, 
% the disparities in accuracy of outcomes resulting from
% decision making with privacy-protected data and
% proposed novel measures of fairness.

\begin{figure}
    \centering
    \includegraphics[width=.99\linewidth]{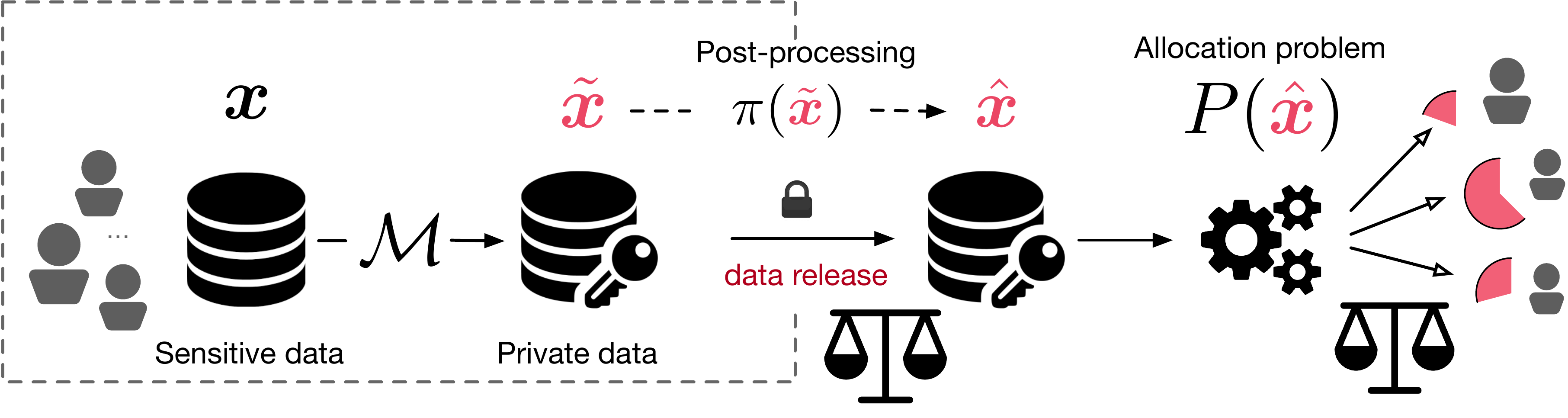}
    \caption{Schematic problem representation.}
    \label{fig:scheme}
\end{figure}

%%%%%%%%%%%%%%%%%%%%%%%%%%%%%%%%%%%%%%%%%%%%%%%%%%%%%%%%%%%%%%%%%%%%%
\section{Preliminaries: Differential Privacy}
%%%%%%%%%%%%%%%%%%%%%%%%%%%%%%%%%%%%%%%%%%%%%%%%%%%%%%%%%%%%%%%%%%%%%

\emph{Differential Privacy} \citep{Dwork:06} (DP) characterizes the amount of individual data disclosed in a computation.

\begin{definition}%[Differential Privacy \citep{Dwork:06}]
  A randomized algorithm $\cM:\cX \to \cR$ with domain $\cX$ and range $\cR$ satisfies
  $\epsilon$-\emph{differential privacy} if
  for any output $O \subseteq \cR$ and datasets $\bm{x}, \bm{x}' \in \cX$ differing by at most one entry (written
  as $\bm{x} \sim \bm{x}'$) 
  \begin{equation}
  \label{eq:dp}
    \Pr[\cM(\bm{x}) \in O] \leq \exp(\epsilon) \Pr[\cM(\bm{x}') \in O] + \delta. 
  \end{equation}
\end{definition}

\noindent 
Parameter $\epsilon \!>\! 0$ is the \emph{privacy loss}: values close 
to $0$ denote strong privacy and $\delta \geq 0$ represents a probability of failure. Intuitively, DP states that 
every event has a similar probability regardless of the participation
of any individual data to the dataset. 
% individual data is added or removed to the dataset, limiting the 
% amount of information that the output reveals about any individual.  
DP satisfies several properties including 
% \emph{composition}, which allows to bound the privacy loss derived by 
% multiple applications of DP algorithms to the same dataset, and 
\emph{immunity to post-processing}, which states that the privacy 
loss of DP outputs is not affected by arbitrary data-independent 
post-processing \citep{Dwork:13}.

A function $f$ from a dataset $\bm{x} \in \cX$ to a result set 
$R \subseteq \RR^n$ can be made differentially private by injecting 
random noise onto its output. The amount of noise relies on the notion 
of \emph{global sensitivity} %, denoted by $\Delta_f$ and defined as
\(
\Delta_f = \max_{\bm{x} \sim \bm{x}'} \| f(\bm{x}) - f(\bm{x}') \|_p
\) with $p \in \{1,2\}$. 
The \emph{Laplace mechanism} \citep{Dwork:06} that outputs $f(\bm{x}) 
+ \bm{\eta}$, where $\bm{\eta} \in \RR^n$ is drawn from the i.i.d.~Laplace distribution with $0$ mean and scale  
$\Delta_f/\epsilon$ over $n$ dimensions, achieves $\epsilon$-DP. 
The \emph{Gaussian mechanism} \citep{Dwork:13} that outputs $f(D) + \bm{\eta}$, 
where $\bm{\eta} \in \RR^n$ is drawn from the multivariate normal distribution 
$\cN(\bm{0},\sigma^2 \bm{I}_n)$ with parameter $\sigma\geq c\Delta_f/\epsilon$, 
achieves $(\epsilon,\delta)$-differential privacy, for $c^2>2\ln(1.25/\delta)$.

\section{Settings and Goals}
\label{sec:settings}
%%%%%%%%%%%%%%%%%%%%%%%%%%%%%%%%%%%%%%%%%%%%%%%%%%%%%%%%%%%%%%%%%%%%%%%%%%%%%%%%%%%%%%%%%%%%%%

The paper considers datasets $\bm{x} \!\in\! \RR$ of $n$ entities, 
whose elements $x_i$ describe some measurable quantities of entity 
$i \!\in\! [n]$, such as the number of individuals living 
in a geographical region $i$. A data-release mechanism $\cM$ is applied to the dataset $\bm{x}$ (called true data in this paper) to produce a privacy-preserving counterpart $\noisydata \sim \cM(\truedata)$ (referred to as noisy data). Given the released data, the paper considers allocation problems $P: \RR^n \to \mathbb{R}^n$ that distribute a finite set of resources to the problem entities. For example, $P$ may be used to allocate funds to school districts. 

The focus of the paper is to study the error disparities of a DP data-release mechanism $\cM$ in two contexts: {\bf (1)} data release and {\bf (2)} downstream decisions. The first context refers to the case in which the noisy data must be post-processed before being released to satisfy desired invariants. The second context refers to the case in which the noisy data is released for use in an allocation problem. Again, the release data must be post-processed to satisfy the problem-specific feasibility constraints. {\em The paper studies the disparate impacts of the error introduced by post-processing among entities in both two scenarios.} 

% \keyu{
% The paper analyzes the error disparities resulting from a DP mechanism $\cM$ in two scenarios: {\bf (1)} the noisy data $\noisydata$ gets post-processed before release;
% and {\bf (2)} the noisy data $\noisydata$ gets released directly.
% Figure \ref{fig:scheme} provides an illustrative
% diagram. 

% In the first scenario, the released data has already met feasibility requirements and thus can
% be directly used for the downstream decision processes. However, the post-processing mechanism
% applied during the data-release period might result in the disparities in accuracy for both
% the released data and the decision processes dependent on it. In this case, the paper
% mainly focuses on the post-processing method used for the data-release task.
% In the second scenario, since random noise is injected to the true data $\truedata$, the output 
% $P(\noisydata)$ might not necessarily satisfy intrinsic feasibility constraints 
% of $P(\truedata)$ (e.g., non-negativity), which makes post-processing a necessity. 
% However, even with proper post-processing, the error with respective to $P(\truedata)$
% might still be inevitable and different among entities. }

Quantitatively, this error is represented by the bias associated with a 
post-processing mechanism $\pp$, i.e.,
\begin{equation*}
    \bias{\pp,P}{\cM,\truedata}=\EE{\noisydata\sim\cM(\truedata)}{\pp(\noisydata)}-P(\truedata)\,.
\end{equation*}
The paper will often omit the last two arguments of the bias term when there is no ambiguity. 
% $\bias{\pp,P}{\cM,\truedata}=\bias{\pp}{P}$.
% \keyu{
% Let $\norm{\cdot}_\rightleftharpoons$ denote an operator, which computes the maximum difference among
% different entries of input, i.e., $\norm{\bm{v}}_{\rightleftharpoons} = \max_{i\in[n]}~v_i-\min_{i\in [n]}~v_i$.
% }
The disparate impact of the error is then characterized by the following definition.
\begin{definition}[$\alpha$-fairness]
    A post-processing mechanism $\pp$ is said $\alpha$-fair with respect to  problem
    $P$ if the maximum difference among the biases is bounded by $\alpha$, i.e.,
    \begin{equation*}
        \norm{\bias{\pp}{P}}_{\rightleftharpoons} = \max_{i\in[n]}~\bias{\pp}{P}_i-\min_{i\in[n]}~\bias{\pp}{P}_i\leq \alpha\,
    \end{equation*}
\end{definition}
\noindent
with $\alpha$ referred to as a \emph{fairness bound} that captures the fairness violation.

% The results in the paper assume that $\cM$, used to release counts, 
% is the Laplace or Gaussian mechanism with an appropriate finite sensitivity $\Delta$.
% \emph{However, the results are general and apply to any data-release
%   DP mechanism that add unbiased noise}.

\begin{figure*}[!h]
    \centering
    \includegraphics[width=0.4\linewidth]{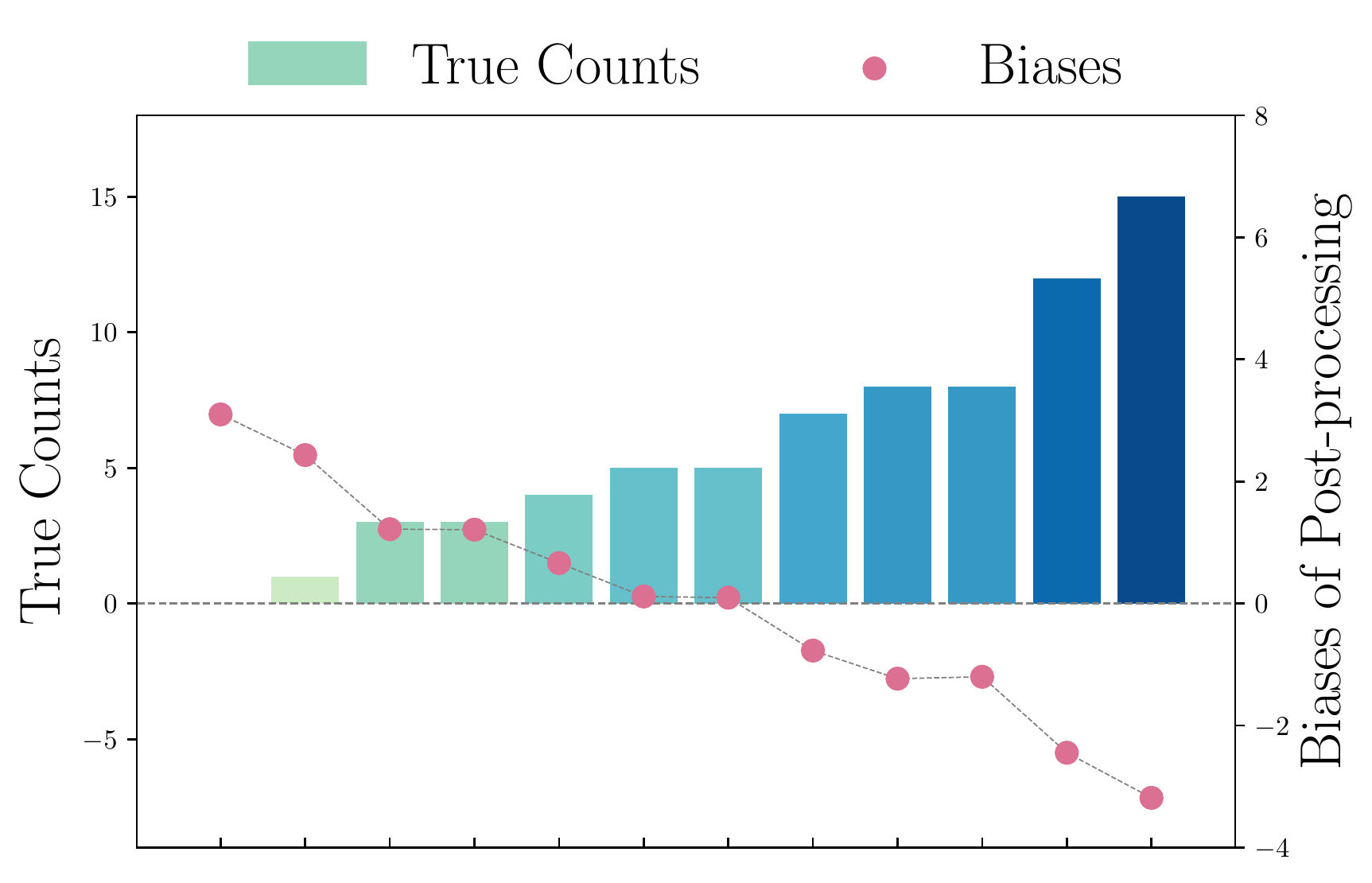} \hspace{12pt}
    \includegraphics[width=0.4\linewidth]{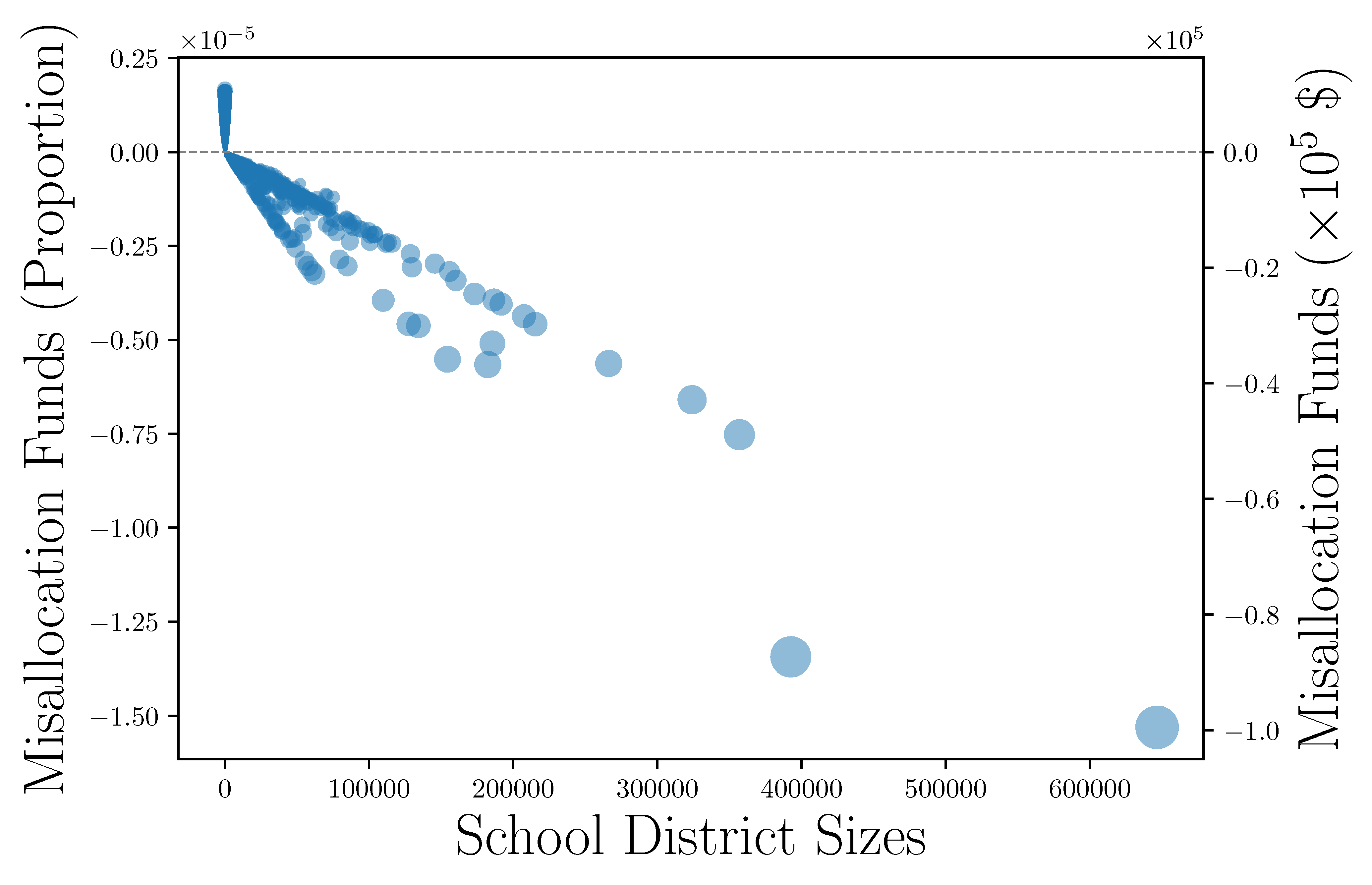}

    \caption{{\bf Data Release}: (Left) bar chart of the true counts and line chart of the empirical biases  (red dots)
    associated with the given post-processing mechanism.
    {\bf Downstream decisions}: (Right) scatter plot of the empirical biases (left y-axis) and misallocation funds (right y-axis) resulting from the given post-processing mechanism. For both two instances,
    Laplace mechanism is taken for privacy protection and 
    each experiment is repeated for 200,000 times.}
    \label{fig:bias_of_post_processing}
\end{figure*}
\section{Motivating Applications}
\label{sec:motivation}

This section reviews two settings highlighting the 
disparate impacts of DP post-processing in census releases.

\paragraph{Data Release} 
Consider a simplified version of the census data release problem. The task is to 
release counts, such as demographic information of individuals, which are required to be non-negative and summed up to a public quantity. The latter is usually used to preserve known statistics at a state or national level. 
To preserve these invariants, commonly adopted post-processing mechanisms (e.g., the one used by the Top-Down algorithm in the release of several 2020 U.S.~census data products) constrain the noisy DP outcomes with an $\ell_2$ \emph{projection} step. Such post-processing step will be described in detail and studied in the next section. Figure \ref{fig:bias_of_post_processing} (left) shows the (sorted) values of some synthetically generated true counts (bars) and the (empirical) biases (red dots), obtained by comparing the post-processed DP counts with the true counterparts. 
Notice how the resulting biases vary among entities. {\em This is significant as sub-communities may be substantially under- or over-counted affecting some important data analysis tasks.}

\paragraph{Downstream Decisions}
These disparities may also have negative socio-economic impacts. 
For instance, when agencies allocate funds and benefits according to 
differentially private data, an ill-chosen post-processing mechanism can result in huge disparities and, as a consequence, lead to significant inefficiencies of allocation. 
Consider the \emph{Title I of the Elementary and Secondary Education Act of
  1965} \citep{Sonnenberg:16}: It uses the US Census data to distribute about \$6.5 billion in basic grants to qualified school districts in proportion to the count $x_i$
of children aged 5 to 17 who
live in necessitous families in district $i$. The allocation is formalized by
\begin{equation*}
    \talloci{i} \coloneqq \frac{a_i\cdot x_i}{\sum_{j=1}^n a_j\cdot
    x_j}\,,\qquad\forall~i\in[n]\,,
\end{equation*}
where $\bm{x}=[x_1~\dots~x_n]^\top$ is the vector of the districts' true counts and $a_i$ is a positive weight factor reflecting students expenditures in district $i$. When a projection mechanism (described in more details in Section \ref{sec:alloc}) is used to guarantee non-negativity of the private data $\noisydata$, the resulting errors on the proposal of funds allocation can be consequential. 
Figure \ref{fig:bias_of_post_processing} (right) visualizes the misallocation (blue dots) for over 16,000 school districts (due to this post-processing mechanism) in 
terms of proportions (left y-axis) and funds (right y-axis). 
In this numerical simulation, which uses data based on the 2010 US census release, {\em large school districts} may be strongly penalized. For example the largest school district in Los Angeles can receive up to 99,000 dollars fewer than warranted.

The next sections analyze these effects and propose mitigating solutions. Due to space limitations, complete proofs are deferred to the Appendix.
\noindent

%%%%%%%%%%%%%%%%%%%%%%%%%%%%%%%%%%%%%%%%%%%%%%%%%%%%%%%%%%%%%%%%%%%%%%%%%%%%%%%%
\section{Unfairness in Data Release}
\label{sec:census}
%%%%%%%%%%%%%%%%%%%%%%%%%%%%%%%%%%%%%%%%%%%%%%%%%%%%%%%%%%%%%%%%%%%%%%%%%%%%%%%
This section studies the effects of post-processing in a common data-release setting, 
where the goal is to release population counts 
that must also sum up to a public constant $C$. 
The section first introduces the projection mechanisms used to restore non-negativity and other aggregate data invariants and then studies its fairness effects. 
% Consistent with the 2020 US Census data-release requirements \citep{abowd2018us}, the released private data is required to be non-negative and to sum up to a constant $C$ (to retain known invariances about aggregate population counts). 
% Since the noisy data $\noisydata$ might fail to meet these two feasibility constraints, data-release agencies use post-processing to restore feasibility. 
% \nandoside{Maybe the above is redundant given the motivating application}

\emph{Projections} are common post-processing methods central to many data-release applications, including energy
\citep{fioretto:TSG20}, transportation \citep{fioretto:AAMAS-18}, and census data
\citep{abowd2019census}. They are defined as:
\begin{equation}
    \begin{aligned}
        \nnsumpostdata(\noisydata) \coloneqq~\underset{\bm{v}\in \nnsumregion}{\arg\min}~& \norm{\bm{v}-\noisydata}_2\,,
        % \nnsumregion=\left\{\bm{v}\in\RR^n~\bigg{|}~\sum_{i=1}^n v_i = b,~\bm{v}\geq \bm{0}\right\},
    \end{aligned}\tag{$P_{\mathrm{S}+}$} \label{nnsumprogram}
\end{equation}
with feasible region defined as 
\[
\nnsumregion=\Big\{\bm{v}\mid\sum_{i=1}^n v_i = C\,,~\bm{v}\geq \bm{0}\Big\}.
\]
Notice that $P_{S+}$ is a convex program, and its unique optimal solution $\pi_{S+}(\tilde{\bm{x}})$ guarantees the desired data invariants by definition.
For the analysis, it is also useful to consider a modified version $P_S$ of \ref{nnsumprogram}, which differs from the latter only in that it ignores the non-negativity constraint $\bm{v} \geq \bm{0}$. Its feasible region and optimal solution are denoted, respectively, $\sumregion$ and  $\sumpostdata(\noisydata)$. 

This section provides tight upper and lower bounds of the unfairness  arising from projection operators. Lemma \ref{thm:bias_ineq} and \ref{thm:bias_ineq_upper} are critical components to derive the $\alpha$-fairness bounds developed in Theorem \ref{thm:bounds}.
The tightness of the proposed bounds is demonstrated in Example \ref{eg:centroid} and the existence of inherent unfairness in
Example \ref{cor:existence_of_unfairness}. Proposition \ref{thm:lower_explicit} then presents an efficient method to evaluate the $\alpha$-fairness bounds under the Gaussian mechanism, giving a uniquely valuable tool to decision makers to evaluate the impact of post-processing the data in their applications. To ease notation, the section omits the second argument $P$ of the bias term $\cB$ (as the $P$ is an identity function for data-release settings). Additionally, unless otherwise specified, it assumes that the noisy data $\noisydata$ is an output of the Laplace mechanism with parameter $\lambda$ or the Gaussian mechanism with parameter $\sigma$. 

\begin{lemma}\citep{zhu2021bias} \label{cor:err_dist}
    For any noisy data $\noisydata\in \RR^n$, the closed-form solution $\sumpostdata(\noisydata)$ to program
    ($P_S$) can be expressed as,
    \begin{equation*}
        \sumpostdata(\noisydata)_i = 
         x_i+\eta_i-\frac{\sum_{j=1}^n\eta_j}{n}
         =\noisydataelem_i + \frac{C-\sum_{j=1}^n
         \noisydataelem_j}{n}\,,
    \end{equation*}
    for any $i\in[n]$, with injected noise $\noise=\noisydata-\truedata$.
\end{lemma}
\noindent
%It provides a close-form expression of the post-processed count $\sumpostdata(\noisydata)$. 
% Additionally, note that the post-processing mechanism $\sumpostdata$ itself does not introduce bias, i.e., $\biasm{\sumpostdata}=\EE{\noisydata}{\sumpostdata(\noisydata)}-\truedata=\bm{0}$
% \citep{zhu2021bias}.

Unlike $\sumpostdata(\noisydata)$, the post-processed count
$\nnsumpostdata(\noisydata)$ does not have a close-form expression. However,
the following lemma introduces an implicit expression of 
$\nnsumpostdata(\noisydata)$ on the basis of $\sumpostdata(\noisydata)$, establishing the foundation for the fairness analysis of the post-processing mechanism $\nnsumpostdata$.
\begin{lemma}\label{lem:kkt_ps}
    For any noisy data $\noisydata\in \RR^n$, the solution $\nnsumpostdata(\noisydata)$ to program \eqref{nnsumprogram}
    can be expressed as
    \begin{equation*}
         \nnsumpostdata(\noisydata) = \relu{\sumpostdata(\noisydata) - T(\sumpostdata(\noisydata))\cdot\bm{1} }\,,
    \end{equation*}
    where 
    $\relu{\cdot} \!=\! \max\{\cdot, 0\}$, and
    $T(\sumpostdata(\tilde{\bm{x}}))$ is the non-negative scalar that is the unique solution to the following equation
    \begin{equation*}
        \sum_{i=1}^n \relu{\sumpostdata(\noisydata)_i - T(\sumpostdata(\noisydata))} = C\,.
    \end{equation*}
\end{lemma}

    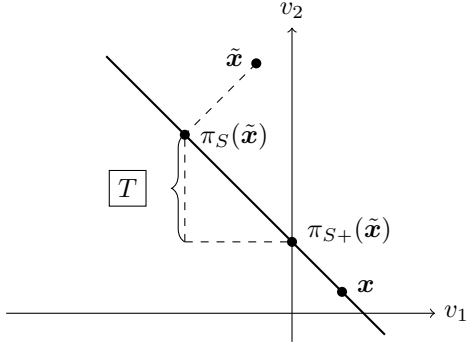
\begin{figure}[t]
    \centering
    %  \resizebox{0.75\linewidth}{!}{%
     \begin{tikzpicture}[scale=.95]
        \draw[->] (-4,0) -- (2,0) node[right] {\normalsize $v_1$};
        \draw[->] (0,-0.4) -- (0,4) node[above] {\normalsize $v_2$};
        \draw[thick] (-2.6, 3.6) -- (1.3, -0.3);
        
        \coordinate (truedata) at (0.7, 0.3);
        \coordinate (noisydata) at (-0.5, 3.5);
        \coordinate (postdata) at (-1.5, 2.5);
        \coordinate (nnpostdata) at (0, 1);
        \coordinate (T) at (-1.5, 1);
        \node at (-2.3,1.75) [rectangle,draw] (v) {\normalsize $T$};
        
        \fill (truedata) circle (2pt) node[right, xshift=.5ex, yshift=.5ex] {\normalsize $\truedata$};
        \fill (noisydata) circle (2pt) node[left,  xshift=-.3ex, yshift=.5ex] {\normalsize $\noisydata$};
        \fill (nnpostdata) circle (2pt) node[right, xshift=0.5ex, yshift=1ex] {\normalsize $\nnsumpostdata(\noisydata)$};
        \fill (postdata) circle (2pt) node[right, xshift=0.5ex, yshift=-0.1ex] {\normalsize $\sumpostdata(\noisydata)$} ;
        \fill (T) circle (0pt) node[left, xshift=-.5ex, yshift=-.75ex] {};
        
        % \draw[dashed] (noisydata) -- (postdata);
        % \draw[dashed] (refpostdata) -- (postdata);
        \draw[dashed] (T) -- (nnpostdata);
        \draw[dashed] (T) -- (postdata);
        \draw[dashed] (noisydata) -- (postdata);
        
        \draw [decorate,decoration={brace,amplitude=5pt},xshift=0pt,yshift=0pt]
        (T) -- (postdata) node [black,midway,xshift=-.5cm] 
        {};
                
    \end{tikzpicture}
    % }
    \caption{Illustration of different post-processed counts of $\noisydata$.
    The solid line represents the feasible region $\sumregion$ of program ($P_S$).}
    \label{fig:err_sum_0}
    \end{figure}

\noindent
Figure \ref{fig:err_sum_0} provides an illustrative example relating the two post-processing mechanisms and the role of $T(\sumpostdata(\tilde{\bm{x}}))$. Given the noisy data $\noisydata$ (top of the figure)  $\sumpostdata$ first projects it onto the solid line, which is the feasible region $\sumregion$. Then,
$T$ needs to be deducted from both two entries of $\sumpostdata(\noisydata)$
such that the positive part of $\sumpostdata(\noisydata)-T\cdot \bm{1}$ equals $\nnsumpostdata(\noisydata)$.

The following lemma provides lower and upper bounds for the bias difference 
of post-processing $\nnsumpostdata$  and plays a critical role in establishing the main results.
\begin{lemma}\label{thm:bias_ineq}
    % Suppose that the noisy data $\noisydata$ is an output of the Laplace mechanism
    % with parameter $\lambda$ or the Gaussian mechanism with parameter $\sigma$.
    For any pair $(i,j)$ such that $x_i\leq x_j$,
    the following relations hold,
    {\small
    \begin{subequations}
    \begin{align}
        \label{eq:bias_ineq}
        \biasm{\nnsumpostdata}_i-\biasm{\nnsumpostdata}_j &\geq
        \biasm{\relu{\sumpostdata}}_i-
        \biasm{\relu{\sumpostdata}}_j\\
        \label{eq:2b}
         \biasm{\nnsumpostdata}_i-\biasm{\nnsumpostdata}_j &\leq 
         \biasm{\relu{\sumpostdata}}_i-
        \biasm{\relu{\sumpostdata}}_j \\
        &\quad +\EE{\sumpostdata(\noisydata)}{T(\sumpostdata(\noisydata))}\notag
    \end{align}
    \end{subequations}
    }
    with $T$ defined as in Lemma \ref{lem:kkt_ps} and $\biasm{\relu{\sumpostdata}}$ used as shorthand for
    $\EE{\noisydata}{\relu{\sumpostdata(\noisydata)}}-\truedata$.
\end{lemma}
\noindent 
While important, the upper bound \eqref{eq:2b} is dependent on function $T$, which does not have a close-form expression; this makes it difficult to evaluate it. The following proposition provides an
upper bound of $T$ using $\sumpostdata(\noisydata)$.
\begin{proposition}\label{prop:bound_T}
    For any noisy data $\noisydata\in\RR^n$, $T(\sumpostdata(\noisydata))$ is upper bounded by the sum of negative parts in $\sumpostdata(\noisydata)$:
    \begin{equation*}
        T(\sumpostdata(\noisydata))\leq \sum_{i=1}^n \negpart{\sumpostdata(\noisydata)_i}\,,
    \end{equation*}
    where $\negpart{\cdot}\!=\!-\min\{\cdot, 0\}$ takes the negative part of the input.
\end{proposition}
\noindent
The following lemma presents an upper bound of the difference between biases: unlike the bound developed in Lemma \ref{thm:bias_ineq}, this new bound is independent of the injected noise.
\begin{lemma}\label{thm:bias_ineq_upper}
    %  Suppose that the noisy data $\noisydata$ is an output of the Laplace mechanism
    % with parameter $\lambda$ or the Gaussian mechanism with parameter $\sigma$. 
    For any pair $(i,j)$ such that $x_i\leq x_j$,
     the following relation holds.
     \begin{equation}\label{eq:bias_ineq_upper}
         \biasm{\nnsumpostdata}_i-\biasm{\nnsumpostdata}_j
         \leq x_j - x_i\,.
     \end{equation}
\end{lemma}
\noindent

The next theorem is the main result of this section: it bounds the unfairness resulting from the projection mechanism $\nnsumpostdata$. 
Without loss of generality, the true data $\truedata$ is assumed to be sorted in an increasing order, i.e., $x_i\leq x_j$, for any $i<j$.
\begin{theorem}\label{thm:bounds}
 The fairness bound $\alpha$ associated with
    the post-processed mechanism $\nnsumpostdata$
    is bounded from the below by
    \begin{equation*}
        \alpha\geq \biasm{\relu{\sumpostdata}}_1-\biasm{\relu{\sumpostdata}}_n\,,
    \end{equation*}
    and bounded from the above by
    \begin{multline*}
        \alpha\leq \min\{x_n-x_1,\biasm{\relu{\sumpostdata}}_1-\\\biasm{\relu{\sumpostdata}}_n
        +\sum_{i=1}^n\EE{\sumpostdata(\noisydata)}{\negpart{\sumpostdata(\noisydata)_i}}\}\,.
    \end{multline*}
\end{theorem}

\begin{proof}[Proof Sketch]
    By Equation \eqref{eq:bias_ineq} in Lemma \ref{thm:bias_ineq}, notice that $\biasm{\nnsumpostdata}_1$ is the largest entry while $\biasm{\nnsumpostdata}_n$
    is the smallest one among all the biases. The lower bound of the fairness bound $\alpha$ can then be derived in the following way.
    \begin{align*}
        \alpha  &=\max_{j\in[n]}~\biasm{\nnsumpostdata}_j - \min_{j\in[n]}~\biasm{\nnsumpostdata}_j\\
        &= \biasm{\nnsumpostdata}_1-\biasm{\nnsumpostdata}_n\geq \biasm{\relu{\sumpostdata}}_1 - \biasm{\relu{\sumpostdata}}_n\,.
    \end{align*}
    Likewise, Lemma \ref{thm:bias_ineq} and \ref{thm:bias_ineq_upper}, along with Proposition \ref{prop:bound_T},
    make the joint effort to generate the upper bound.
\end{proof}
\noindent
The tightness of the derived bounds follows from the following instance.
\begin{example}[Centroid]\label{eg:centroid}
    The lower and upper bounds proposed in Theorem \ref{thm:bounds} hold with equality
    when the true data $\truedata$ is exactly the centroid of the feasible region $\nnsumregion$
    of program \eqref{nnsumprogram}, i.e., $\truedata=[C/n~\dots~C/n]\in\RR^n$, and
    the noisy data $\noisydata$ is an output of either Laplace or Gaussian mechanism. In this case,
    the fairness bound $\alpha$ and its bounds in Theorem \ref{thm:bounds} happen to be $0$, which also means that there is no fairness violation.
    % which demonstrates the \emph{tightness} of
    % the proposed bounds.
\end{example}
\noindent
The next example shows that post-processing definitely introduces unfairness when the true is not at the centroid.

\begin{example}[Non-centroid]\label{cor:existence_of_unfairness}
    Suppose that the true data $\truedata$ is not the centroid of the feasible region $\nnsumregion$, i.e.,
    $x_n > x_1$. The fairness bound $\alpha$ associated with the post-processing mechanism $\nnsumpostdata$
    is strictly positive, i.e.,
    \begin{equation*}
        \alpha\geq \biasm{\relu{\sumpostdata}}_1-\biasm{\relu{\sumpostdata}}_n>0\,.
    \end{equation*}
\end{example}

\noindent
This negative result motivates the development of novel post-processing mechanisms in downstream decision processes, which are topics of the next section. The last result of this section provides an efficient evaluation of the proposed bounds via numerical integration methods.

\begin{proposition}\label{thm:lower_explicit}
    Let $\noisydata$ be the output of the Gaussian mechanism with parameter $\sigma$. 
    The key component of both lower and upper bounds in Theorem \ref{thm:bounds} can be written as
    \begin{multline*}
        \biasm{\relu{\sumpostdata}}_1-\biasm{\relu{\sumpostdata}}_n = \int_{-x_n}^{-x_1} \Phi(at)dt \\
        \in\left[\Phi\left(- a x_n\right)(x_n-x_1),\Phi\left(- a x_1\right)(x_n-x_1)\right] \,,
    \end{multline*}
    where $a = \frac{1}{\sigma}\sqrt{\frac{n}{n-1}}$, and $\Phi(\cdot)$ is the 
    standard Gaussian cumulative distribution function.
\end{proposition}

\noindent
It is interesting to demonstrate the tightness of these bounds using the US Census households counts at the county level for the state of Hawaii. 

\begin{example}[Hawaii]
     The state of Hawaii has a total number of $C=453,558$
    households distributed in $n=5$ counties. The experiments use the Laplace mechanism with parameter $\lambda=10$ and the Gaussian mechanism with parameter
    $\sigma=25$. 
    The empirical studies of $\alpha$-fairness and its bounds in Theorem \ref{thm:bounds} associated with
    the post-processing mechanism $\nnsumpostdata$ over $1,000,000$ independent
    runs are reported in Table \ref{tab:fair_and_bounds}. The bounds of Gaussian mechanism use Proposition \ref{thm:lower_explicit};
    those of Laplace mechanism are generated by
    the empirical means. 
\end{example}

The derived lower and upper bounds are really tight and provide decision makers a uniquely valuable tool to assess the unfairness introduced by post-processing.

\begin{table}
  \centering
  \resizebox{.75\columnwidth}{!}{
    \begin{tabular}{r|rrr}
    \toprule
    \multicolumn{1}{r|}{Mechanism $\cM$} & $\alpha$-fairness & Lower & Upper \\
    \midrule
    Laplace & 0.0245  & 0.0242  & 0.0288  \\
    \midrule
    Gaussian & 0.0910  & 0.0897  & 0.1085  \\
    \bottomrule
    \end{tabular}%
    }
    \caption{Case study of Hawaii.}
  \label{tab:fair_and_bounds}%
\end{table}%

%%%%%%%%%%%%%%%%%%%%%%%%%%%%%%%%%%%%%%%%%%%%%%%%%%%%%%%%%%%%%%%%%%%%%%%%%%%%%%%%%
\section{Mechanisms for Downstream Decisions}
\label{sec:alloc}
%%%%%%%%%%%%%%%%%%%%%%%%%%%%%%%%%%%%%%%%%%%%%%%%%%%%%%%%%%%%%%%%%%%%%%%%%%%%%%%%%

Having shown that unfairness is unavoidable in common data-release settings, %this section notice that as shown in \citep{} these disparate errors can be further exacerbated when processed as input to downstream decision problems. Therefore, 
this section aims at designing  post-processing mechanisms for decision processes that minimize their fairness impact on the resulting decisions. The mechanisms studied are tailored for the allocation problem $\al$ described in Section \ref{sec:motivation}, which captures a wide class of resource allocation problems. 

A natural baseline, currently adopted in census data-release tasks, is to first post-process the noisy data to meet the feasibility requirement (i.e., non-negativity) and then apply the allocation formula $\al$ to the post-processed counts. To restore feasibility, it suffices to take the positive part of $\noisydata$ to obtain $\relu{\noisydata}$, or equivalently, project $\noisydata$ onto the non-negative orthant $\RR^n_+$. 

\begin{definition}[Baseline Mechanism (BL)]\label{def:bl}
     The \emph{baseline mechanism} outputs, for each $i\in [n]$,
     \begin{equation*}
         \blmech{\noisydata}_i \coloneqq \frac{a_i\cdot \relu{\tilde{x}_i}}{\sum_{j=1}^n a_j\cdot \relu{\tilde{x}_j}}\,.
     \end{equation*}
 \end{definition}

\noindent
It is possible to derive results similar to Example 
\ref{cor:existence_of_unfairness} for $\relu{\cdot}$ when the  baseline mechanism is used to produce feasible released data.
Additionally, as shown in \citep{DBLP:conf/ijcai/0007FHY21}, the disparate errors resulting from $\relu{\cdot}$ can be further exacerbated when they are used as inputs to downstream decision problems. It suggests that the baseline mechanism might not be a good candidate for mitigating unfairness in this allocation problem. To address this limitation, consider the optimal post-processing mechanism in this context, i.e.,
\begin{equation}\label{eq:opt}
    \pp^*\coloneqq\min_{\pp \in\Pi_{\Delta_n}} ~ \norm{\EE{\noisydata}{\pp(\noisydata)-\truealloc}}_{\rightleftharpoons}\,,
\end{equation}
where $\Pi_{\Delta_n}=\left\{\pp\mid \pp:\RR^n\mapsto \Delta_n\right\}$ represents a class of post-processing mechanisms whose images belong to the probability simplex $\Delta_n$. The optimization problem in Equation \eqref{eq:opt} is intractable in its direct form, since $\truealloc$ is not available to the mechanism,  motivating the need to approximate the objective function. Consider the following proxy
$\EE{\noisydata}{\norm{\pp(\noisydata)-\noisyalloc}_{\rightleftharpoons}}$, which first exchanges the order of expectation and $\norm{\cdot}_{\rightleftharpoons}$ and
then replaces the true allocation $\truealloc$ with its noisy variant $\noisyalloc$.
Then, the optimal post-processing mechanism $\pp_\alpha^*$ associated with this new proxy function becomes:
\begin{equation}
    \begin{aligned}
    \pp_{\alpha}^*(\noisydata)\coloneqq\underset{\bm{v}\in\Delta_n}{\arg\min}  ~
    \norm{\bm{v}-\noisyalloc}_{\rightleftharpoons}
    % \mathrm{s.t.}&~\sum_{i=1}^n v_i=1,~\bm{v}\geq \bm{0}\,.
\end{aligned}\tag{$P_{\alpha}$}\label{prog:s}
\end{equation}
\noindent
A mechanism, which is closely related to program \eqref{prog:s}, is presented as follows.
 \begin{definition}[Projection onto Simplex Mechanism (PoS)]\label{def:pos} 
     The \emph{projection onto simplex mechanism} outputs the allocation as follows.
     \begin{equation*}
        \begin{aligned}
            \posmech{\noisydata} \coloneqq~\underset{\bm{v}\in \Delta^n}{\arg\min}~& \norm{\bm{v}-\noisyalloc}_2
            % \qquad
            % \mathrm{s.t. }~& \sum_{i=1}^n v_i = 1\,,~\bm{v}\geq \bm{0}\,.
            % \nnsumregion=\left\{\bm{v}\in\RR^n~\bigg{|}~\sum_{i=1}^n v_i = b,~\bm{v}\geq \bm{0}\right\},
        \end{aligned}\tag{$P_{\mathrm{PoS}}$}\label{prog:pos}
\end{equation*}
\end{definition}

Program \eqref{prog:pos} projects $\noisyalloc$, which is not necessarily an allocation since it may violate non-negativity constraints, onto the closest feasible allocation. The next theorem establishes the equivalence between program \eqref{prog:s} and
program \eqref{prog:pos}: It leads to a near-optimal post-processing mechanism. (The missing proofs of the rest of this paper can be found in Appendix \ref{app:proof_of_sec_alloc}).

\begin{theorem}\label{thm:a-fair}
    For any noisy data $\noisydata$, the mechanism $\posmech{\noisydata}$ generates
    the unique optimal solution to program \eqref{prog:s}.
\end{theorem}

\begin{figure}
    \centering
    \includegraphics[width=\linewidth]{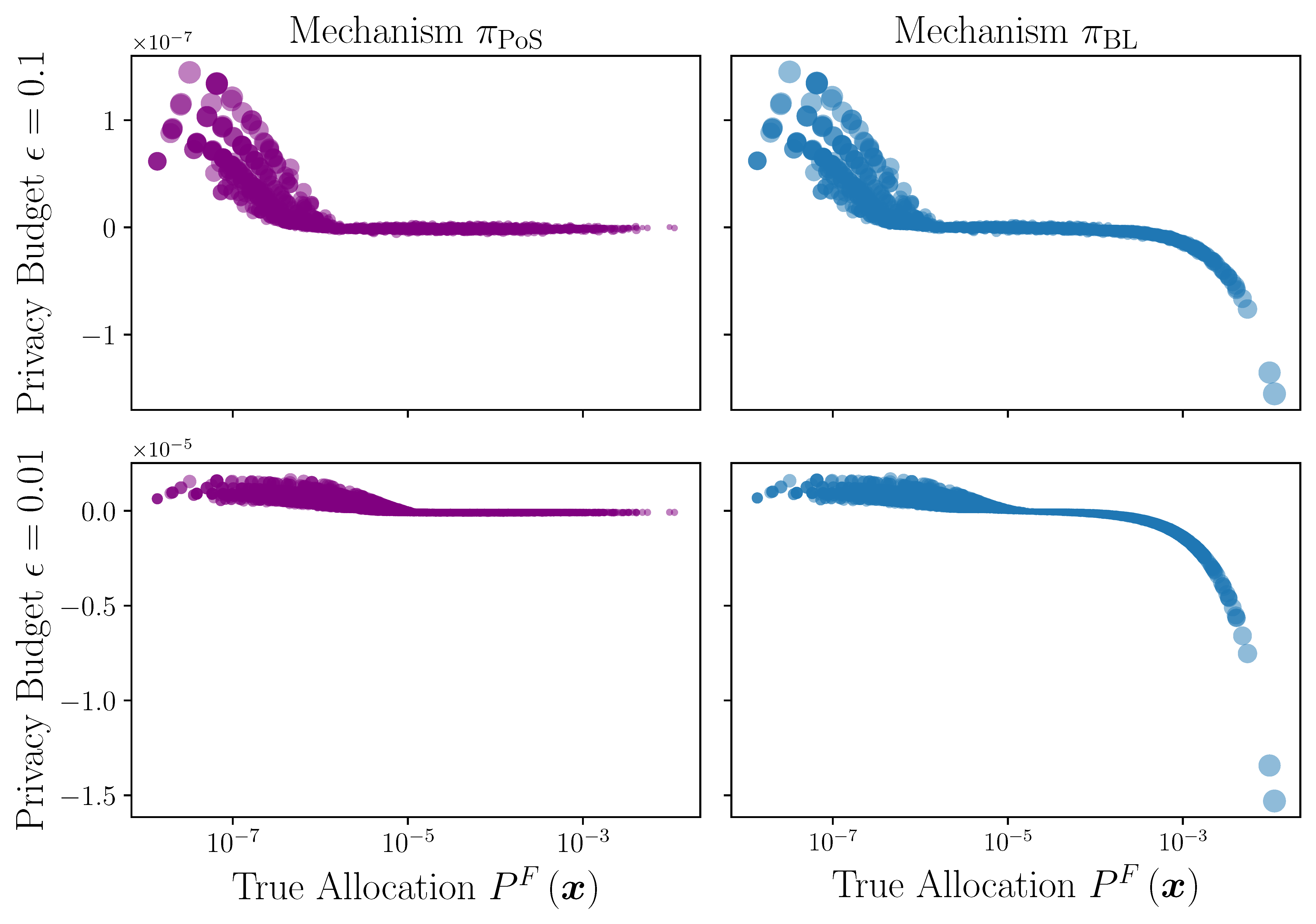}
    \caption{Illustration of the empirical biases ($y$-axis) associated with the two mechanisms $\pos$ and $\bl$
    (columns) for different privacy budgets (rows)
    versus the portions of education funds ($x$-axis) schools districts are guaranteed in the allocation with the
    true data. The Laplace mechanism is used for privacy protection and each experiment is repeated for
    200,000 times.}
    \label{fig:bias}
\end{figure}

\begin{table*}[!t]
  \centering
  \resizebox{0.65\linewidth}{!}{
    \begin{tabular}{l|rr|rr|rr}
    \toprule
    Privacy Budgets & \multicolumn{2}{c|}{$\epsilon=0.1$} & \multicolumn{2}{c|}{$\epsilon=0.01$} & \multicolumn{2}{c}{$\epsilon=0.001$} \\
    \midrule
    Mechanisms & $\bl$    & $\pos$   &  $\bl$     & $\pos$   &  $\bl$     
    & $\pos$ \\
    \midrule
    $\alpha$-fairness & 3.00E-07 & \textbf{1.50E-07} & 1.70E-05 & \textbf{1.75E-06} & 8.06E-04 & \textbf{2.23E-05} \\
    Cost of Privacy & 1.62E-05 & \textbf{1.41E-05} & 1.33E-03 & \textbf{1.04E-03} & 5.90E-02 & \textbf{3.49E-02} \\
    \bottomrule
    \end{tabular}%
    }
  \caption{Comparison between the two post-processing mechanisms in terms of 
  two fairness metrics for different privacy budgets. This work takes Laplace mechanism and 
  $200,000$ independent runs for numerical evaluation.}
  \label{tab:fairness_metrics}%
\end{table*}%

\noindent
Figure \ref{fig:bias}  visualizes the resulting biases
of the Title I allocation associated with these two mechanisms, $\pos$ and $\bl$. It is noteworthy that these two mechanisms achieve roughly same performance for the school districts that are allocated small amounts. However, under the baseline mechanism, the school districts that account for a significant portion of total budget receive much less funding than what they are supposed to receive when no differential privacy is applied. This is not the case for mechanism
$\pos$, which reduces unfairness significantly. Recall that the notion of $\alpha$-fairness measures the maximum difference among
biases associated with different entities. Pictorially, the biases associated with $\pos$ do not vary as drastically as the baseline mechanism. Table \ref{tab:fairness_metrics} quantifies the benefits 
of  $\pos$ over $\bl$.

%%%%%%%%%%%%%%%%%%%%%%%%%%%%%%%%%%%%%%%%%%%%%%%%
\section{Generalizations}
\label{sec:generalization}
%%%%%%%%%%%%%%%%%%%%%%%%%%%%%%%%%%%%%%%%%%%%%%%%

The results in Section \ref{sec:alloc} can be generalized to
other fairness metrics. This section discusses an important metric that quantifies the extra budget needed to ensure that all of the problem entities receive the resources (e.g., amounts of funds) they are warranted by law. 

\begin{definition}[Cost of privacy \citep{DBLP:conf/ijcai/0007FHY21}]
    Given the mechanism $\pp$, the total budget  $B$ to distribute and the true data
    $\truedata$,
    the \emph{cost of privacy} is defined as
    \begin{equation*}
        \cop\coloneqq \sum_{j\in \cJ^-} \vert \bias{\pp}{\al}_j\vert \cdot B\,,
    \end{equation*}
    with the index set $\cJ^-\coloneqq \left\{j\mid\bias{\pp}{\al}_j<0\right\}$.
\end{definition}

\noindent The next proposition establishes the equivalence between the cost of privacy and the 
$\ell_1$ norm of the bias when the image of the mechanism $\pp$ is restricted to
be the probability simplex $\Delta_n$.
\begin{proposition}[Cost of privacy as a $\ell_1$-norm]
\label{prop:cop_as_l1_norm}
    Suppose that $\pp$ is a post-processing mechanism, which belongs to the class $\Pi_{\Delta_n}$.
    The cost of privacy is a multiplier of the $\ell_1$-norm
    of its bias, i.e.,
    \begin{equation*}
        \cop = \frac{B}{2}\cdot\norm{\bias{\pp}{\al}}_1\,.
    \end{equation*}
\end{proposition}
\noindent
Since the optimal post-processing is again intractable in its direct form, i.e., it cannot be solved as an optimization problem, its objective can be replaced by the proxy $B/2\cdot\EE{\noisydata}{\norm{\pp(\noisydata)-\noisyalloc}_1}$.
Then, the optimal post-processing mechanism $\pp_{\mathrm{CoP}}^*$ associated with this proxy function is given by
\begin{equation}
    \begin{aligned}
    \pp_{\mathrm{CoP}}^*(\noisydata)\coloneqq\underset{\bm{v}\in\Delta_n}{\arg\min}  ~
    \frac{B}{2}\cdot\norm{\bm{v}-\noisyalloc}_1\,.
    % \mathrm{s.t.}&~\sum_{i=1}^n v_i=1,~\bm{v}\geq \bm{0}\,;
\end{aligned}\tag{$P_{\mathrm{CoP}}$}\label{prog:s-cop}
\end{equation}
The next theorem depicts the connection between $P_{\mathrm{CoP}}$
and the two post-processing mechanisms proposed in Section \ref{sec:alloc}.
\begin{theorem}\label{thm:cop}
    For any noisy data $\noisydata$, the mechanism $\posmech{\noisydata}$ generates
    an optimal solution to program \eqref{prog:s-cop}. For any noisy data $\noisydata$ 
    such that $\sum_{j=1}^n a_j\cdot \tilde{x}_j >0$, 
    mechanism $\blmech{\noisydata}$ generates
    an optimal solution to program \eqref{prog:s-cop} as well.
\end{theorem}
\noindent
This theorem demonstrates that mechanism $\pos$ always produces
an optimal solution to program \eqref{prog:s-cop} while the baseline mechanism achieves
optimality with high probability. Table \ref{tab:fairness_metrics} shows that $\pos$ may significantly outperform the baseline mechanism, providing substantial reductions in the cost of privacy.

\section{Discussion and Conclusion}
This paper was motivated by the recognition that the disparate error impacts of post-processing of differentially private outputs are poorly understood. Motivated by Census applications, it took a first step toward understanding how and why post-processing may produce disparate errors in data release and  downstream allocation tasks. 
The paper showed that a popular class of post-processing mechanisms commonly adopted to restore invariants during the release of population statistics are inherently unfair. It proposed a tight bound on the unfairness and discussed an efficient method to evaluate the disparate impacts. 
Motivated by these negative results, the paper studied how post-processed data affects downstream decisions under a fairness lens and how to contrast such effects. In this context, the paper proposed to release the noisy, non-post-processed data, and post-processing the output of the downstream decisions instead.  It focused on an important class of resource allocation problems used to allot funds or benefits and proposed a novel (approximately) optimal post-processing mechanism that is effective in mitigating unfairness under different fairness metrics. 
The analysis was complemented with numerical simulation on funds allocation based on private Census data showing up to an order magnitude improvements on different accuracy disparity metrics. 

These results may have strong implications with respect to fairness in downstream decisions and should inform statistical agencies about the advantage of releasing private non-post-processed data, in favor of designing post-processing methods directly applicable in the downstream decision tasks of interest. 

\section*{Acknowledgement}
This research is partially supported by the National Science Foundation (NSF 2133169 and NSF 2133284). The opinions expressed are solely those of the authors.

\bibliographystyle{named}
\bibliography{ijcai22}

%%%%%%%%%%%%%%%%%%%%%%%%%%%%%%%%%%%%%%%%%%%%%%%%%%%%%%%%%%%%%%%%%%%%%%
\appendix
%%%%%%%%%%%%%%%%%%%%%%%%%%%%%%%%%%%%%%%%%%%%%%%%%%%%%%%%%%%%%%%%%%%%%%
\begin{table*}[!th]
     \centering
    %\resizebox{\linewidth}{!} 
    % \vspace{2pt}
     {%%
     \begin{tabular}{cl}
     \toprule
        $\truedata$ & True data\\
      $\noisydata$ & Noisy data with noise injected: $\noisydata=\truedata+\bm{\eta}$ \\
    %   $\weight$ & The weight vector of the $n$ districts associated with students expenditures\\
      $n$  & Dimension of the true data $\truedata$ \\
    %   $\postdata$ & Post-processed solution of program \eqref{program} with linear constraints $\bm{A}\bm{v}=\bm{b}$.\\
    %   $\nnpostdata$ & Post-processed solution of program \eqref{nnprogram} with linear constraints $\bm{A}\bm{v}=\bm{b}$ and \\
    %   & non-negativity constraints $\bm{v}\geq \bm{0}$.\\
    $\relu{}$ & Operator taking the positive part of the input, i.e., $\relu{\bm{v}}=\max\{\bm{v},\bm{0}\}$\\
    $\negpart{}$ & Operator taking the negative part of the input, i.e., $\relu{\bm{v}}=-\min\{\bm{v},\bm{0}\}$\\
    $\norm{\cdot}_{\rightleftharpoons}$ & Operator taking the difference between maximum and minimum value of the input, \\
    &i.e.,
    $\norm{\bm{v}}_{\rightleftharpoons}=\max_{i\in[n]}v_i - \min_{i\in[n]}v_i$.\\
      $\sumpostdata$ & Post-processed solution of program $P_S$ with summation constraint $\sum_{i=1}^n\bm{v}_i = C$\\
      $\sumregion$ & Feasible region of program $P_S$  $\left\{\bm{v}\mid\sum_{i=1}^n\bm{v}_i = C\right\}$\\
      $\relu{\sumpostdata}$ & Post-processed solution of program $P_S$, taking the non-negative part only\\
      $\nnsumpostdata$ & Post-processed solution of program \eqref{nnsumprogram} with summation constraint $\sum_{i=1}^n\bm{v}_i = C$ and \\
      $\nnsumregion$ & Feasible region of program \eqref{nnsumprogram}  $\left\{\bm{v}\mid\sum_{i=1}^n\bm{v}_i = C, \bm{v}\geq \bm{0}\right\}$\\
      $\cM$ & Differentially private mechanism for noise addition \\
      $\al$ & allocation formula.\\
      $\talloci{i}$ & The true allocation for entity $i$\\
      $\nalloci{i}$ & The noisy allocation for entity $i$\\
      $B$ & Total budget in allocation problem\\
      \bottomrule
     \end{tabular}
     }
     \caption{Important notations adopted in this paper}
     \label{tab:symbols}
\end{table*}
\noindent
The commonly adopted notation and symbols throughout the paper and the appendix are reported in Table \ref{tab:symbols}.

%%%%%%%%%%%%%%%%%%%%%%%%%%%%%%%%%%%%%%%%%%%%%%%%%%%%%%%%%%%%%%%%%%%%%%%%%%%%%%%%%%%%%%%%%%%%%%
\section{Missing Proofs}
%%%%%%%%%%%%%%%%%%%%%%%%%%%%%%%%%%%%%%%%%%%%%%%%%%%%%%%%%%%%%%%%%%%%%%%%%%%%%%%%%%%%%%%%%%%%%%

\subsection{Projection-related Results}
\label{app:proj}

In the first place, this section introduces a general equality-constrained post-processing
mechanism \eqref{program}
\begin{equation}
     \begin{aligned}
        \postdata(\noisydata)\coloneqq \underset{\bm{v}\in\RR^n}{\arg\min}&~\norm{\bm{v}-\noisydata}_2\\
        \mathrm{s.t.}&~
    \bm{A}\bm{v}=\bm{b}\,.
     \end{aligned}\tag{$P$}\label{program}
\end{equation}
A general version of Lemma \ref{cor:err_dist} associated with program \eqref{program} is presented
as follows.
\begin{lemma}\label{lem:kkt}
    For any noisy data $\noisydata\in \RR^n$, the closed-form solution $\postdata(\noisydata)$ to program \eqref{program}
    can be expressed as
    \begin{equation*}
         \postdata(\noisydata) = \truedata + (\bm{I}_n - \bm{A}^\top(\bm{A}\bm{A}^\top)^{-1}\bm{A})(\noisydata-\truedata)\,,
    \end{equation*}
    where $\bm{I}_n$ is the identity matrix of size $n$.
\end{lemma}

\begin{proof}
    Consider the following convex optimization problem, which is equivalent to program \eqref{program}.
    \begin{equation*}
        \begin{aligned}
            \postdata(\noisydata)=\underset{\bm{v}\in\RR^n}{\arg\min}~
            \frac{1}{2}\norm{\bm{v}-\noisydata}_2^2 \qquad
            \text{s.t.}~  \bm{A}\bm{v}=\bm{b}\,.
        \end{aligned}
    \end{equation*}
    The Lagrange function is then given by
    \begin{equation*}
        L(\bm{v},\bm{\mu})=\frac{1}{2}\inner{\bm{v}-\noisydata}{\bm{v}-\noisydata}+\inner{\bm{\mu}}{\bm{A}\bm{v}-\bm{b}}\,.
    \end{equation*}
    To solve program \eqref{program} in exact form, it suffices to find a feasible solution $(\bm{v}^*, \bm{\mu}^*)$ to 
    the Karush-Kuhn-Tucker (KKT) conditions shown as follows.
    \begin{align*}
        \bm{v}^*-\noisydata +  \bm{A}^\top\bm{\mu}^* =&~~\bm{0}\,, & &(\textrm{Stationarity})\\
        \bm{A}\bm{v}^*=&~~\bm{b}\,. & &(\textrm{Primal feasibility})
    \end{align*}
    It follows that
    \begin{equation*}
        \bm{A}\bm{v}^*=\bm{A}(\noisydata- \bm{A}^\top\bm{\mu}^*)=\bm{A}\noisydata- \bm{A}\bm{A}^\top\bm{\mu}^*=\bm{b}\,.
    \end{equation*}
    Since $\bm{A}$ is of full row rank, the $m\times m$ matrix $\bm{A}\bm{A}^\top$ is invertible, which implies that
    \begin{align}
        \bm{\mu}^* &= (\bm{A}\bm{A}^\top)^{-1}(\bm{A}\noisydata - \bm{b})\,,\nonumber\\
        \bm{v}^* &= \noisydata - \bm{A}^\top \bm{\mu}^*=\noisydata - \bm{A}^\top(\bm{A}\bm{A}^\top)^{-1}(\bm{A}\noisydata - \bm{b})\nonumber\\
        &=\truedata + (\noisydata-\truedata) - \bm{A}^\top(\bm{A}\bm{A}^\top)^{-1}(\bm{A}\noisydata - A\truedata)\label{eq:feasible_x}\\
        &=\truedata + (\bm{I}_n - \bm{A}^\top(\bm{A}\bm{A}^\top)^{-1}\bm{A})(\noisydata-\truedata)\,,\nonumber
    \end{align}
    where Equation \eqref{eq:feasible_x} comes from the assumption that the true data $\truedata$ belongs to
    $\region$. Therefore, $\bm{v}^*$ is an optimal solution to program \eqref{program} and equal to
    $\postdata(\noisydata)$ due to uniqueness of $\postdata(\noisydata)$.
\end{proof}
\noindent
Likewise, a more general non-negative equality-constrained post-processing mechanism
\eqref{nnprogram} is defined as follows.
\begin{equation}
          \begin{aligned}
             \nnpostdata(\noisydata)\coloneqq \underset{\bm{v}\in\RR^n}{\arg\min}&~\norm{\bm{v}-\noisydata}_2\\
             \mathrm{s.t.}&~
         \bm{A}\bm{v}=\bm{b}\,,\\
         &~\bm{v}\geq \bm{0}\,.
          \end{aligned}\tag{$P_+$}\label{nnprogram}
\end{equation}
Let $\cL=\{\bm{v}\mid \bm{A}\bm{v}=\bm{0}\}$ represent the linear subspace associated with the affine subspace
$\region=\{\bm{v}\mid \bm{A}\bm{v}=\bm{b}\}$ and $\cL^{\bot}$ the orthogonal complement of $\cL$. 
It follows that the $n$-dimensional linear space $\RR^n$ is the direct sum of the linear subspace $\cL$
and its orthogonal complement $\cL^\bot$, i.e., $\RR^n=\cL\oplus \cL^\bot$. For notational convenience, let $n'$ denote the dimension
of the linear subspace $\cL$, $\dim(\cL)$. Since the matrix $\bm{A}$ is of full row rank, $\dim(\cL)=n'=n-m$ and $\dim\left(\cL^\bot\right)=m$.

\begin{lemma}\label{lem:how_to_project}
    For any $\bm{y}\in \region$,
    $\noisydata \in \bm{y}+\cL^\bot$ if and only if $\postdata(\noisydata) = \bm{y}$, i.e.,
    \begin{equation*}
        \bm{y}+\cL^\bot = \left\{\noisydata\mid \postdata(\noisydata)=\bm{y}\right\}\,.
    \end{equation*}
\end{lemma}
\begin{proof}
    (``only if" part) Suppose that there exists some $\bm{l}\in\cL^\bot$ and
    $\bm{\gamma}\in\RR^m$ such that $\noisydata=\bm{y}+\bm{l}$ and $\bm{l}=\bm{A}^\top \bm{\gamma}$. By Lemma \ref{lem:kkt}, it follows that 
    \begin{align}
        \postdata(\noisydata)&=\truedata + (\bm{I}_n - \bm{A}^\top(\bm{A}\bm{A}^\top)^{-1}\bm{A})(\noisydata-\truedata)
        \nonumber\\
        &=\truedata + (\bm{I}_n -
        \bm{A}^\top(\bm{A}\bm{A}^\top)^{-1}\bm{A})\left[\left(\bm{y}-\truedata\right)+\bm{A}^\top\bm{\gamma}\right]
        \nonumber\\
        &=\truedata+\left(\bm{y}-\truedata\right)+(\bm{I}_n -
        \bm{A}^\top(\bm{A}\bm{A}^\top)^{-1}\bm{A})\bm{A}^\top\bm{\gamma}\label{eq:how_to_project_1}\\
        &=\bm{y}+(\bm{A}^\top - \bm{A}^\top(\bm{A}\bm{A}^\top)^{-1}\bm{A}\bm{A}^\top)\bm{\gamma}\nonumber\\
        &=\bm{y}+(\bm{A}^\top - \bm{A}^\top )\bm{\gamma}\nonumber\\
        &=\bm{y}\nonumber\,,
    \end{align}
    where Equation \eqref{eq:how_to_project_1} is due to the fact that $\bm{y}-\truedata\in \cL$ such that $\bm{A}(\bm{y}-\truedata)=\bm{0}$ and thus
    \begin{align*}
        &(\bm{I}_n - \bm{A}^\top(\bm{A}\bm{A}^\top)^{-1}\bm{A})(\bm{y}-\truedata)\\
        =~&(\bm{y}-\truedata)+\bm{A}^\top(\bm{A}\bm{A}^\top)^{-1}\left[\bm{A}(\bm{y}-\truedata)\right]\\
        =~&\bm{y}-\truedata\,.
    \end{align*}
    (``if" part) Suppose that the noisy data $\noisydata$ does not belong to $\bm{y}+\cL^\bot$. Since $\region$
    is a shift of $\cL$ and $\RR^n=\cL\oplus \cL^\bot$, there exists a unique
    representation of $\noisydata$, which is  $\noisydata=\bm{u}+\bm{l}\in \bm{u}+\cL^\bot$
    where $\bm{u}\in \region$ and $\bm{l}\in\cL^\bot$. It follows that
    $\postdata(\noisydata)=\bm{u}=\bm{y}$ and, thus, $\noisydata\in\bm{y}+\cL^\bot$,
    which is contradictory to the assumption.
\end{proof}
\noindent
The following corollary is a direct consequence of Lemma \ref{lem:how_to_project} because, for 
program $(P_S)$, the affine subspace $\cL$ is $\{\bm{v}\mid \bm{1}^\top\bm{v}=0\}$ and
its orthogonal component $\cL^\bot$ is $\{k\cdot \bm{1}\mid k\in\RR\}$.
\begin{corollary}\label{cor:how_to_project}
    For any $\bm{y}\in \sumregion$,
    $\sumpostdata(\noisydata) = \bm{y}$ if and only if $\noisydata \in \bm{y}+\cL^\bot$, i.e.,
    \begin{equation*}
        \left\{\noisydata\mid \sumpostdata(\noisydata)=\bm{y}\right\}=\bm{y}+\cL^\bot =\left\{\bm{y}+k\cdot \bm{1}\mid k\in\RR\right\}\,.
    \end{equation*}
\end{corollary}
\begin{lemma}\label{lem:proj_equal_proj_proj}
    For any $\noisydata\in\RR^n$, program \eqref{nnprogram} yields the same solution for both $\noisydata$
    and $\postdata(\noisydata)$, i.e.,
    \begin{equation*}
        \nnpostdata(\noisydata)=\nnpostdata(\postdata(\noisydata))\,.
    \end{equation*}
\end{lemma}
\begin{proof}
    It suffices to show the following inequality.
    \begin{equation*}
        \norm{\noisydata-\nnpostdata(\postdata(\noisydata))}_2\leq \norm{\noisydata-\nnpostdata(\noisydata)}_2\,.
    \end{equation*}
    Note that, for any $\bm{u}\in \nnregion=\left\{\bm{v}\mid \bm{A}\bm{v}=\bm{b},\bm{v}\geq \bm{0}\right\}$,
    \begin{equation}\label{eq:lem_d3_1}
        \norm{\postdata(\noisydata)-\nnpostdata(\postdata(\noisydata))}_2\leq \norm{\postdata(\noisydata)-\bm{u}}_2\,,
    \end{equation}
    and $\postdata(\noisydata)-\bm{u}\in \cL$ because both $\postdata(\noisydata)$ and
    $\bm{u}$ belong to $\region$. Also, by Lemma \ref{lem:how_to_project}, there exists some $\bm{l}\in\cL^\bot$
    such that $\noisydata=\postdata(\noisydata)+\bm{l}$. Then, it follows that, for any $\bm{u}\in\nnregion$,
    \begin{align}
        \norm{\noisydata-\bm{u}}_2^2&=\norm{\left(\noisydata-\postdata(\noisydata)\right)+\left(\postdata(\noisydata)-\bm{u}\right)}_2^2\nonumber\\
        &=\norm{\bm{l}+\left(\postdata(\noisydata)-\bm{u}\right)}_2^2\nonumber\\
        &=\norm{\bm{l}}_2^2+\norm{\postdata(\noisydata)-\bm{u}}_2^2+2\inner{\bm{l}}{\postdata(\noisydata)-\bm{u}}\nonumber\\
        &=\norm{\bm{l}}_2^2+\norm{\postdata(\noisydata)-\bm{u}}_2^2\label{eq:lem_d3_2}\\
        &\geq \norm{\bm{l}}_2^2 + \norm{\postdata(\noisydata)-\nnpostdata(\postdata(\noisydata))}_2^2\label{eq:lem_d3_3}\\
        &= \norm{\noisydata-\postdata(\noisydata)}_2^2 + \norm{\postdata(\noisydata)-\nnpostdata(\postdata(\noisydata))}_2^2\nonumber\\
        &=\norm{\noisydata-\nnpostdata(\postdata(\noisydata))}_2^2\,,\label{eq:lem_d3_4}
    \end{align}
    where Equation \eqref{eq:lem_d3_2} and \eqref{eq:lem_d3_4} come from the fact that
    $\left(\noisydata-\postdata(\noisydata)\right)~\bot~ 
    \left(\postdata(\noisydata)-\bm{u}\right)$ for any $\bm{u}\in\nnregion$ and
    $\nnpostdata(\postdata(\noisydata))\in\nnregion$. Inequality 
    \eqref{eq:lem_d3_3} is a direct consequence of \eqref{eq:lem_d3_1}. Since
    $\nnpostdata(\noisydata)\in \nnregion$,
    by taking $\bm{u}=\nnpostdata(\noisydata)$, Equation \eqref{eq:lem_d3_4} becomes
    \begin{equation*}
        \norm{\noisydata-\nnpostdata(\postdata(\noisydata))}_2\leq
        \norm{\noisydata-\nnpostdata(\noisydata)}_2\,,
    \end{equation*}
    which completes the proof here.
\end{proof}

\subsection{Auxiliary Lemmas }
Let $\swap{i}{j}{\bm{v}}$ denote the swapping operator, which swaps the values of
the $i$-th and $j$-th entries of the vector $\bm{v}$
while leaving the rest unchanged. Algebraically, $\swap{i}{j}{\bm{v}}$
is to multiply a permutation matrix $\bm{P}_{(ij)}$ times the vector $\bm{v}$, i.e., $\swap{i}{j}{\bm{v}}=\bm{P}_{(ij)}\bm{v}$.

\begin{lemma}\label{lem:swap_invariance}
   For any noisy data $\noisydata\in\RR^n$ and any pair $(i,j)$ such that $i\neq j$,
   \begin{align*}
       \swap{i}{j}{\sumpostdata(\noisydata)}&=\sumpostdata\left(\swap{i}{j}{\noisydata}\right)\,,\\
       \swap{i}{j}{\nnsumpostdata(\noisydata)}&=
       \nnsumpostdata\left(\swap{i}{j}{\noisydata}\right)\,,\\
      T(\sumpostdata(\noisydata))&=T\left(\sumpostdata\left(\swap{i}{j}{\noisydata}\right)\right)\,.
   \end{align*}
\end{lemma}

 \begin{lemma}\label{lem:prob_comp_E}
    Suppose that the noisy data $\noisydata$ follows a multivariate normal distribution 
    $\cN(\truedata,\sigma^2 \bm{I}_n)$.
    For any vector $\bm{y}\in \region$, the probability density of the post-processed solution
    $\postdata(\noisydata)$ at $\bm{y}$ is given by
    \begin{equation*}
        f_{\postdata(\noisydata)}(\bm{y}) = 
        \frac{1}{\sqrt{(2\pi)^{n'}n'}}\exp\left(-\frac{1}{2}\norm{\bm{y}-\truedata}_2^2\right)\,,
    \end{equation*}
    where $n'$ is the affine dimension of $\region$.
\end{lemma}

\begin{corollary}\label{cor:pdf_diff}
    Suppose that the noisy data $\noisydata$ is generated by the Gaussian mechanism with
    parameter $\sigma$.
    For any pair $(i,j)$ and vector $\bm{y}\in \sumregion$ such that $x_i\leq x_j$ and 
    $  y_i\leq   y_j$, 
    the probability density of the post-processed solution
    $\sumpostdata(\noisydata)$ at $\bm{y}$ is no less than that at its swapping counterpart $\swap{i}{j}{\bm{y}}$,
    i.e.,
    \begin{equation*}
        f_{\sumpostdata(\noisydata)}(\bm{y})\geq f_{\sumpostdata(\noisydata)}(\swap{i}{j}{\bm{y}})\,.
    \end{equation*}
\end{corollary}

\begin{proof}
    By Lemma \ref{lem:prob_comp_E}, the probability density of the post-processed solution
    $\sumpostdata(\noisydata)$ at $\bm{y}$ and $\swap{i}{j}{\bm{y}}$ is given by the following formulas respectively
    \begin{equation*}
        f_{\sumpostdata(\noisydata)}(\bm{y}) = 
        \frac{1}{\sqrt{(2\pi)^{n'}n'}}\exp\left(-\frac{1}{2}\norm{\bm{y}-\truedata}_2^2\right)>0\,,
    \end{equation*}
    and
    \begin{multline*}
        f_{\sumpostdata(\noisydata)}(\swap{i}{j}{\bm{y}}) =\\ 
        \frac{1}{\sqrt{(2\pi)^{n'}n'}}\exp\left(-\frac{1}{2}\norm{\swap{i}{j}{\bm{y}}-\truedata}_2^2\right)>0\,.
    \end{multline*}
    Notice that the ratio of $f_{\sumpostdata(\noisydata)}(\bm{y})$ to
    $f_{\sumpostdata(\noisydata)}(\swap{i}{j}{\bm{y}})$
    can be expressed as
    \begin{align}
        &\frac{f_{\sumpostdata(\noisydata)}(\bm{y})}{f_{\sumpostdata(\noisydata)}(\swap{i}{j}{\bm{y}})}\nonumber\\
        =~&
        \exp\left(\frac{1}{2}\left(\norm{\swap{i}{j}{\bm{y}}-\truedata}_2^2-\norm{\bm{y}-\truedata}_2^2 \right)\right)\nonumber\\
        % =~&\exp\left(\frac{1}{2}\left((  y_j-  x_i)^2+(  y_i-  x_j)^2-(  y_i-  x_i)^2
        % -(  y_j-  x_j)^2\right)\right)\nonumber\\
        =~&\exp\left(\frac{1}{2}\left(-2  y_j  x_i-2  y_i  x_j+2  y_i  x_i+2  y_j
          x_j\right)\right)\nonumber\\
        =~&\exp\left((  x_i-  x_j)(  y_i-  y_j)\right)\label{eq:cor_pdf_diff_1}\\
        \geq~& 1\,,\nonumber
    \end{align}
    where Equation \eqref{eq:cor_pdf_diff_1} comes from the assumption that $  x_i$ is no greater than
    $  x_j$ for any $x_i\leq x_j$. It implies the following, for any $\bm{y}\in\sumregion$ such that $y_i\leq y_j$,
    \begin{equation*}
        f_{\sumpostdata(\noisydata)}(\bm{y})\geq f_{\sumpostdata(\noisydata)}(\swap{i}{j}{\bm{y}})\,.
    \end{equation*}
\end{proof}

\begin{lemma}\label{lem:lap_pdf_diff}
   Suppose that the noisy data $\noisydata$ is generated by the Laplace mechanism
   with parameter $\lambda$.
    For any pair $(i,j)$ and vector $\bm{y}\in \sumregion$ such that $x_i\leq x_j$ and $y_i\leq y_j$, 
    the probability density of the post-processed solution
    $\sumpostdata(\noisydata)$ at $\bm{y}$ is no less than that at its swapping counterpart $\swap{i}{j}{\bm{y}}$,
    i.e.,
    \begin{equation*}
         f_{\sumpostdata(\noisydata)}(\bm{y})\geq f_{\sumpostdata(\noisydata)}(\swap{i}{j}{\bm{y}})\,.
    \end{equation*}
\end{lemma}

\begin{proof}
    Note that the probability density function of the noisy data $\noisydata$ 
    can be expressed as
    \begin{equation*}
        f_{\noisydata}(\bm{v}) = \frac{\lambda^n}{2^n}\exp\left(-\lambda\norm{\bm{v}-\truedata}_1\right)\,,
        \qquad\forall~\bm{v}\in\RR^n\,.
    \end{equation*}
    It follows that the probability density of the post-processed solution $\sumpostdata(\noisydata)$ at
    $\bm{y}$ and $\swap{i}{j}{\noisydata}$ can be given by
    \begin{equation*}
        f_{\sumpostdata(\noisydata)}(\bm{y})=\int_{\bm{v}\in
        \bm{y}+\cL^\bot}\frac{\lambda^n}{2^n}\exp\left(-\lambda\norm{\bm{v}-\truedata}_1\right)d\bm{v}\,,
    \end{equation*}
    and
    \begin{align*}
       &f_{\sumpostdata(\noisydata)}(\swap{i}{j}{\bm{y}})\\
       =~&\int_{\bm{v}\in
        \swap{i}{j}{\bm{y}}+\cL^\bot}\frac{\lambda^n}{2^n}\exp\left(-\lambda\norm{\bm{v}-\truedata}_1\right)d\bm{v}\\
        =~&\int_{\bm{v}\in
        \bm{y}+\cL^\bot}\frac{\lambda^n}{2^n}\exp\left(-\lambda\norm{\bm{v}-\bm{y}+\swap{i}{j}{\bm{y}}-\truedata}_1\right)d\bm{v}\,.
    \end{align*}
    Thus, the difference between these two probability densities, $f_{\sumpostdata(\noisydata)}(\bm{y})$ and
    $f_{\sumpostdata(\noisydata)}(\swap{i}{j}{\bm{y}})$,  is
    \begin{equation}\label{eq:lap_pdf_diff}
        \begin{aligned}
            &\int_{\bm{v}\in
        \bm{y}+\cL^\bot}\frac{\lambda^n}{2^n}\exp\left(-\lambda\norm{\bm{v}-\truedata}_1\right)d\bm{v}-\\
        &
        \int_{\bm{v}\in
        \bm{y}+\cL^\bot}\frac{\lambda^n}{2^n}\exp\left(-\lambda\norm{\bm{v}-\bm{y}+\swap{i}{j}{\bm{y}}-\truedata}_1\right)d\bm{v}\\
        =~&\frac{\lambda^n}{2^n}\int_{\bm{v}\in\bm{y}+\cL^\bot}
        w(\bm{v})d\bm{v}\,,
        \end{aligned}
    \end{equation}
    where 
    \begin{multline*}
        w(\bm{v})=\exp\left(-\lambda\norm{\bm{v}-\truedata}_1\right)\cdot\\\left[
        1-\exp\left(\lambda\norm{\bm{v}-\truedata}_1-
        \lambda\norm{\bm{v}-\bm{y}+\swap{i}{j}{\bm{y}}-\truedata}_1\right)\right]\,.
    \end{multline*}
    In order to conclude the proof here, it suffices to demonstrate that the function $w(\bm{v})$ is non-negative for
    any $\bm{v}\in \bm{y}+\cL^\bot$. By Corollary \ref{cor:how_to_project}, for any $\bm{v}\in\bm{y}+\cL^\bot$,
    there exists $k\in\RR$ such that $\bm{v}=\bm{y}+k\cdot\bm{1}$. A
    critical component of $w(\bm{v})$ is then presented in the following.
    \begin{equation}\label{eq:lap_component}\tag{$*$}
        \begin{aligned}
            &\norm{\bm{v}-\truedata}_1-
            \norm{\bm{v}-\bm{y}+\swap{i}{j}{\bm{y}}-\truedata}_1\\
            =~&\norm{\bm{y}+k\cdot \bm{1}-\truedata}_1-
            \norm{\swap{i}{j}{\bm{y}}+k\cdot \bm{1}-\truedata}_1\\
            =~&\vert y_i+k-  x_i\vert+\vert y_j+k-  x_j\vert-\\
            &\vert y_j+k-  x_i\vert-
            \vert y_i+k-  x_j\vert\,.
        \end{aligned}
    \end{equation}
    Observe that $y_i+k$ is no greater than $y_j+k$ while $  x_i$ is assumed to be no greater than
    $  x_j$ as well. Consider the following two scenarios.
    \begin{enumerate}
        \item If $  x_i\leq y_i+k$, 
        \begin{multline*}
             x_i\leq y_i+k\leq y_j+k\implies \\\vert y_i+k-  x_i\vert-\vert y_j+k-  x_i\vert=-\vert y_i-y_j\vert\,, 
        \end{multline*}
        which implies that
        \begin{multline*}
            \eqref{eq:lap_component}=\vert  y_j+k-  x_j\vert-\vert y_i- y_j\vert-\vert  y_i+k-  x_j\vert\leq\\ \vert  y_j+k-  x_j\vert-\vert  y_j+k-  x_j\vert=0\,.
        \end{multline*}
        \item If $  x_i> y_i+k$, 
        \begin{multline*}
             y_i+k<  x_i\leq x_j\implies\\ \vert  y_i+k-  x_i\vert-\vert  y_i+k-  x_j\vert=-\vert  x_j-  x_i\vert\,,
        \end{multline*}
        which implies that
        \begin{multline*}
            \eqref{eq:lap_component}=\vert  y_j+k-  x_j\vert-\vert  x_j-  x_i\vert-\vert  y_j+k-  x_i\vert\leq \\\vert  y_j+k-  x_j\vert-\vert  y_j+k-  x_j\vert=0\,.
        \end{multline*}
    \end{enumerate}
It follows that, for any $\bm{v}\in\bm{y}+\cL^\bot$,
\begin{equation*}
    w(\bm{v})=\exp\left(-\lambda\norm{\bm{v}-\truedata}_1\right)\left[
        1-\exp\left(
        \lambda\cdot\eqref{eq:lap_component}\right)\right]\geq 0\,.
\end{equation*}
By Equation \eqref{eq:lap_pdf_diff}, the difference between $f_{\sumpostdata(\noisydata)}(\bm{y})$ and $f_{\sumpostdata(\noisydata)}(\swap{i}{j}{\bm{y}})$ turns out to be non-negative, i.e.,
\begin{equation*}
    f_{\sumpostdata(\noisydata)}(\bm{y})-f_{\sumpostdata(\noisydata)}
    (\swap{i}{j}{\bm{y}})=\frac{\lambda^n}{2^n}\int_{\bm{v}\in\bm{y}+\cL^\bot}
    w(\bm{v})d\bm{v}\geq 0\,,
\end{equation*}
which implies that
\begin{equation*}
    f_{\sumpostdata(\noisydata)}(\bm{y})\geq
    f_{\sumpostdata(\noisydata)}(\swap{i}{j}{\bm{y}})\,.
\end{equation*}
\end{proof}

\begin{proposition}\label{prop:func_d}
    The function $d(y,t)$ is decreasing in its first argument.
\end{proposition}
\begin{proof}
    Note that, for any $t\in\RR_+$,
    \begin{equation*}
        d(y,t)=\relu{y-t}-\relu{y}=\begin{cases}
            0\,, & y<0\,,\\
            -y\,, & 0\leq y\leq t\,,\\
            -t\,, & \text{otherwise}\,.
        \end{cases}
    \end{equation*}
    With the formula above, it is straightforward to verify its monotonicity.
\end{proof}

\subsection{Missing Proofs of Results in the Main Text}
\label{app:proof_in_main}

    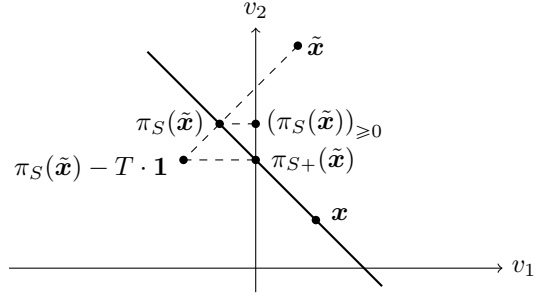
\begin{figure}[t]
    \centering
    %  \resizebox{0.75\linewidth}{!}{%
     \begin{tikzpicture}[scale=.8]
        \draw[->] (-4.1,0) -- (4.1,0) node[right] {\normalsize $v_1$};
        \draw[->] (0,-0.4) -- (0,4) node[above] {\normalsize $v_2$};
        \draw[thick] (-1.8, 3.6) -- (2.1, -0.3);
        
        \coordinate (truedata) at (1, 0.8);
        \coordinate (noisydata) at (0.7, 3.7);
        \coordinate (refpostdata) at (0, 2.4);
        \coordinate (postdata) at (-0.6, 2.4);
        \coordinate (nnpostdata) at (0, 1.8);
        \coordinate (T) at (-1.2, 1.8);
        
        \fill (truedata) circle (2pt) node[right, xshift=.5ex, yshift=.5ex] {\normalsize $\truedata$};
        \fill (noisydata) circle (2pt) node[right] {\normalsize $\noisydata$};
        \fill (nnpostdata) circle (2pt) node[right, xshift=0.5ex, yshift=0ex] {\normalsize $\nnsumpostdata(\noisydata)$};
        \fill (postdata) circle (2pt) node[left, xshift=-0.5ex, yshift=0ex] {\normalsize $\sumpostdata(\noisydata)$} ;
        \fill (refpostdata) circle (2pt) node[right, xshift=.1ex] {\normalsize $\relu{\sumpostdata(\noisydata)}$} ;
        \fill (T) circle (2pt) node[left, xshift=-.5ex, yshift=-.75ex] {\normalsize
        $\sumpostdata(\noisydata)-T\cdot\bm{1}$};
        
        % \draw[dashed] (noisydata) -- (postdata);
        \draw[dashed] (refpostdata) -- (postdata);
        \draw[dashed] (T) -- (nnpostdata);
        \draw[dashed] (noisydata) -- (T);
        
    \end{tikzpicture}
    % }
    \caption{A complete version of Figure \ref{fig:err_sum_0}.}
    \label{fig:full}
    \end{figure}

\subsubsection*{Proof of Lemma \ref{lem:kkt_ps}}
\begin{proof}[Proof of Lemma \ref{lem:kkt_ps}]
    Consider the following convex optimization problem, which is equivalent to program \eqref{nnsumprogram}.
    \begin{equation*}
        \begin{aligned}
            \nnsumpostdata(\noisydata)=\underset{\bm{v}\in\RR^n}{\arg\min}~&
            \frac{1}{2}\norm{\bm{v}-\noisydata}_2^2\\
            \text{s.t.}~  &\inner{\bm{1}}{\bm{v}} =C\,,~\bm{v}\geq \bm{0}\,.
        \end{aligned}
    \end{equation*}
    The Lagrange function is then given by
    \begin{equation*}
        L(\bm{v},\mu,\bm{\omega})=\frac{1}{2}\inner{\bm{v}-\noisydata}{\bm{v}-\noisydata}+\mu\left(\inner{\bm{1}}{\bm{v}}-C\right)-\inner{\bm{\omega}}{\bm{v}}\,.
    \end{equation*}
    To solve program \eqref{nnsumprogram} in exact form, it suffices to find a feasible solution $(\bm{v}^*,\mu^*, \bm{\omega}^*)$ to 
    the Karush-Kuhn-Tucker (KKT) conditions shown as follows.
    \begin{align*}
        \bm{v}^*-\noisydata +  \mu^*\cdot\bm{1} -\bm{\omega}^*=&~~\bm{0}\,, & &(\textrm{Stationarity})\\
        \inner{\bm{1}}{\bm{v}^*}=&~~C\,, & &(\textrm{Primal feasibility})\\
        \bm{v}^*\geq&~~\bm{0}\,, & &(\textrm{Primal feasibility})\\
        \bm{\omega}^*\geq&~~\bm{0}\,, & &(\textrm{Dual feasibility})\\
        \bm{\omega}^*\circ \bm{v}^*=&~~\bm{0}\,. & &(\textrm{Complementary Slackness})
    \end{align*}
    It follows that, for any $i\in [n]$,
    \begin{enumerate}
        \item if $\noisydataelem_i -\mu^* \geq 0$,
        \begin{equation*}
            v^*_i=\noisydataelem_i- \mu^* +\omega^*_i=\noisydataelem_i -\mu^*\,,\qquad \omega^*_i = 0\,;
        \end{equation*}
        \item if $\noisydataelem_i -\mu^* < 0$,
        \begin{equation*}
            v^*_i=\noisydataelem_i- \mu^* +\omega^*_i=0\,,\qquad \omega^*_i = -\noisydataelem_i+\mu^*\,.
        \end{equation*}
    \end{enumerate}
    As a consequence, 
    $\bm{v}^*$ equals the positive part of the difference between $\noisydata$ and $\mu^*\cdot \bm{1}$, i.e.,
    \begin{equation*}
        \bm{v}^*=\relu{\noisydata - \mu^*\cdot \bm{1}}\,.
    \end{equation*}
    Then, according to one of the primal feasibility constraints, 
    \begin{equation*}
        h(\mu^*)\coloneqq\sum_{i=1}^n \relu{\noisydataelem_i - \mu^*}=\inner{\bm{1}}{\bm{v}^*}=C\,,
    \end{equation*}
    the optimal solution $\mu^*$ does exist and turns out to be unique because the function $h:\RR\mapsto\RR_+$
    is decreasing and the constant $C$ is assumed to be positive. Although $\mu^*$ is only expressed in an
    implicit form, it can at least be empirically evaluated via linear search by leveraging  the monotonic
    nature of the function $h$.\\
    Notice that the post-processed solution $\sumpostdata(\noisydata)$ is summed up to the constant $C$
    for any noisy data $\noisydata\in\RR^n$. By Lemma \ref{cor:err_dist}, the following relation holds
    \begin{align*}
        h\left(\frac{C-\sum_{i=1}^n \noisydataelem_i}{n}\right)&=\sum_{i=1}^n\relu{\noisydataelem_i-\frac{C-\sum_{i=1}^n \noisydataelem_i}{n}}\\
        &\geq \sum_{i=1}^n\noisydataelem_i-\frac{C-\sum_{i=1}^n \noisydataelem_i}{n}\\
        &=
        \sum_{i=1}^n \sumpostdata(\noisydata)_i=C=h(\mu^*)\,.
    \end{align*}
    Due to monotonicity of the function $h$, $\mu^*$ is no less 
    than $\left(C-\sum_{i=1}^n \noisydataelem_i\right)/n$. Let $T(\noisydata)$ denote the difference between 
    $\mu^*$ and $\left(C-\sum_{i=1}^n \noisydataelem_i\right)/n$, which proves to be non-negative. The
    optimal solution
    $\bm{v}^*$ can then be expressed as
    \begin{align*}
        \bm{v}^*&=\nnsumpostdata(\noisydata)=\relu{\noisydata - \mu^*\cdot \bm{1}}\\
        &=\relu{\noisydata-
        \left(\frac{C-\sum_{i=1}^n \noisydataelem_i}{n}+T(\noisydata)\right)\cdot
        \bm{1}}\\
        &=\relu{\sumpostdata(\noisydata)
        -T(\noisydata)\cdot \bm{1}}\,.
    \end{align*}
    Moreover, by Lemma \ref{lem:proj_equal_proj_proj}, the post-processed solution $\nnsumpostdata(\noisydata)$
    of program \eqref{nnsumprogram}
    can be expressed as
\begin{align*}
    \nnsumpostdata(\noisydata)&=\nnsumpostdata(\sumpostdata(\noisydata))=
        \relu{\sumpostdata(\sumpostdata(\noisydata))
        -T(\sumpostdata(\noisydata))\cdot \bm{1}}\\
        &=\relu{\sumpostdata(\noisydata)
        -T(\sumpostdata(\noisydata))\cdot \bm{1}}\\
        &=\relu{\sumpostdata(\noisydata)
        -T(\noisydata)\cdot \bm{1}}\,,
\end{align*}
which indicates that $T(\noisydata)$ is identical to $T(\sumpostdata(\noisydata))$.
\end{proof}

\begin{proof}[Proof of Proposition \ref{prop:bound_T}]
    Note that both $\nnsumpostdata(\noisydata)$ and $\sumpostdata(\noisydata)$ belong to 
    the feasible region $\sumregion$. It follows that
    \begin{equation*}
        \sum_{i=1}^n \left(\nnsumpostdata(\noisydata)_i - \sumpostdata(\noisydata)_i\right) = 0\,.
    \end{equation*}
    Recall that Lemma \ref{lem:kkt_ps} shows that 
    \begin{equation*}
        \nnsumpostdata(\noisydata) = \relu{\sumpostdata(\noisydata)-T(\sumpostdata(\noisydata))\cdot \bm{1}}\,,
    \end{equation*}
    which implies that
    \begin{equation*}
        \nnsumpostdata(\noisydata)_i\begin{cases}
            =0\geq \sumpostdata(\noisydata)_i\,, & \sumpostdata(\noisydata)_i\leq 0\,,\\
            \leq \sumpostdata(\noisydata)_i\,, & \text{otherwise}.
        \end{cases}
    \end{equation*}
    Let $\cI$ denote the index set, which is a collection of the negative entries of the post-processed 
    solution of program $(P_S)$, $\sumpostdata(\noisydata)$, i.e.,
    \begin{equation*}
        \cI\coloneqq \left\{i\in[n]\mid \sumpostdata(\noisydata)_i\leq 0\right\}\,.
    \end{equation*}
    Since the constant $C$ is positive, there exists at least one entry $j\in [n]\setminus \cI$ such that
    $\nnsumpostdata(\noisydata)_j$ is positive and, thus, $\sumpostdata(\noisydata)_j$ is positive as
    well. It follows that, for the index $j$, 
    \begin{multline*}
        \nnsumpostdata(\noisydata)_j = \relu{\sumpostdata(\noisydata)_j-T(\sumpostdata(\noisydata))}=\\
        \sumpostdata(\noisydata)_j - T(\sumpostdata(\noisydata))_j>0\,,
    \end{multline*}
    which indicates that
    \begin{align*}
        T(\sumpostdata(\noisydata))&=\sumpostdata(\noisydata)_j-\nnsumpostdata(\noisydata)_j\\
        &\leq
        \sum_{i\in [n]\setminus \cI} \left(\sumpostdata(\noisydata)_i-\nnsumpostdata(\noisydata)_i\right)
        \\&=\sum_{i\in\cI} \left(\nnsumpostdata(\noisydata)_i-\sumpostdata(\noisydata)_i\right)\\
        &=-\sum_{i\in\cI} \sumpostdata(\noisydata)_i=\sum_{i=1}^n \negpart{\sumpostdata(\noisydata)_i}
        \,.
    \end{align*}
\end{proof}

%%%%%%%%%%%%%%%%%%%%%%%%%%%%%%%%%%%%%%%%%%%%%%%%%%%%%%%%%%%%%%%%%%%%%%%%%%%%%%%%%%
% Proof of the lower and upper bounds
%%%%%%%%%%%%%%%%%%%%%%%%%%%%%%%%%%%%%%%%%%%%%%%%%%%%%%%%%%%%%%%%%%%%%%%%%%%%%%%%%%

\subsubsection*{Proof of Lemma \ref{thm:bias_ineq}}
\begin{proof}[Proof of Lemma \ref{thm:bias_ineq}]
    By Lemma \ref{lem:kkt_ps}, the right hand side of Equation \eqref{eq:bias_ineq} can be written as 
    \begin{align*}
         &\biasm{\nnsumpostdata}_i-\biasm{\nnsumpostdata}_j\\
        =~&\left(\EE{\sumpostdata(\noisydata)}{\relu{\sumpostdata(\noisydata)_i - T(\sumpostdata(\noisydata))}}-x_i\right)
        -\\
        &\left(\EE{\sumpostdata(\noisydata)}{\relu{\sumpostdata(\noisydata)_j - T(\sumpostdata(\noisydata))}}
        -x_j\right)\,.
    \end{align*}
    Let $\delta:\sumregion\mapsto\RR$ denote the function
    \begin{multline*}
        \delta(\bm{y})\coloneqq\\\underbrace{\left(\relu{y_i - T(\bm{y})}-\relu{y_i}\right)}_{d(y_i,T(\bm{y}))}-\underbrace{
        \left(\relu{y_j - T(\bm{y})}-\relu{y_j}\right)}_{d(y_j,T(\bm{y}))}
    \end{multline*}
    such that
    \begin{equation}\label{eq:bias_ineq_equiv}
            \begin{aligned}
                &\left(\biasm{\nnsumpostdata}_i-\biasm{\nnsumpostdata}_j\right)-\\
                &\left(
                \biasm{\relu{\sumpostdata}}_i-
                \biasm{\relu{\sumpostdata}}_j\right)\\
                =~&\left(\EE{\sumpostdata(\noisydata)}{\relu{\sumpostdata(\noisydata)_i - T(\sumpostdata(\noisydata))}}-x_i\right)
                -\\
                &\left(\EE{\sumpostdata(\noisydata)}{\relu{\sumpostdata(\noisydata)_j - T(\sumpostdata(\noisydata))}}
                -x_j\right)-\\
                &\left(\EE{\sumpostdata(\noisydata)}{\relu{\sumpostdata(\noisydata)_i}}-x_i\right)+\\
                &\left(\EE{\sumpostdata(\noisydata)}{\relu{\sumpostdata(\noisydata)_j}}-x_j\right)\\
                =~&\EE{\sumpostdata(\noisydata)}{\delta(\sumpostdata(\noisydata))}\,.
            \end{aligned}
    \end{equation}
    Note that the feasible region $\sumregion$ can be partitioned into two
    almost disjoint sets $S_{i,j}^{-}$ and $S_{i,j}^{+}$
    \begin{align*}
        S_{i,j}^{-}&=\sumregion\bigcap H_{i,j}^-=\sumregion\bigcap\left\{\bm{y}\mid
        y_i-y_j\leq 0\right\}\,,\\
        S_{i,j}^{+}&=\sumregion\bigcap H_{i,j}^+=\sumregion\bigcap\left\{\bm{y}\mid
        y_i-y_j\geq 0\right\}\,,
    \end{align*}
    and $\swap{i}{j}{\cdot}$ is a bijection between these two sets. It follows that
    \begin{align}
        &\EE{\sumpostdata(\noisydata)}{\delta(\sumpostdata(\noisydata))}\nonumber\\
        =~&\EE{\sumpostdata(\noisydata)}{d(\sumpostdata(\noisydata)_i,T(\sumpostdata(\noisydata)))
        -d(\sumpostdata(\noisydata)_j,T(\sumpostdata(\noisydata)))}\nonumber\\
        =~&\int_{\bm{y}\in S_{i,j}^{-}}
        (d(y_i,T(\bm{y}))-d(y_j,T(\bm{y})))\cdot f_{\sumpostdata(\noisydata)}(\bm{y})d\bm{y}+\nonumber\\
        &\int_{\bm{y}'\in S_{i,j}^{+}}
        (d(\bm{y}'_i,T(\bm{y}'))-d(\bm{y}'_j,T(\bm{y}')))\cdot f_{\sumpostdata(\noisydata)}(\bm{y}')d\bm{y}'\nonumber\\
        =~&\int_{\bm{y}\in S_{i,j}^{-}}
        (d(y_i,T(\bm{y}))-d(y_j,T(\bm{y})))\cdot f_{\sumpostdata(\noisydata)}(\bm{y})d\bm{y}+\nonumber\\
        &\int_{\bm{y}\in S_{i,j}^{-}}
        (d(y_j,T(\swap{i}{j}{\bm{y}}))-
        d(y_i,T(\swap{i}{j}{\bm{y}})))\cdot\nonumber\\
        &\qquad \quad f_{\sumpostdata(\noisydata)}(\swap{i}{j}{\bm{y}})\vert\det(\bm{P}_{(ij)})\vert d{\bm{y}}
        \label{eq:bias_ineq_long_1}\\
        =~&\int_{\bm{y}\in S_{i,j}^{-}}
        (d(y_i,T(\bm{y}))-d(y_j,T(\bm{y})))\cdot f_{\sumpostdata(\noisydata)}(\bm{y})d\bm{y}+\nonumber\\
        &\quad(d(y_j,T(\bm{y}))-
        d(y_i,T(\bm{y})))\cdot
        f_{\sumpostdata(\noisydata)}(\swap{i}{j}{\bm{y}})d{\bm{y}}\label{eq:bias_ineq_long_2}\\
        =~&\int_{\bm{y}\in S_{i,j}^{-}}
        (d(y_i,T(\bm{y}))-d(y_j,T(\bm{y})))\cdot\nonumber\\
        &\qquad\quad\left(f_{\sumpostdata(\noisydata)}(\bm{y})-f_{\sumpostdata(\noisydata)}(\swap{i}{j}{\bm{y}})\right)d\bm{y}\nonumber\\
        =~&\int_{\bm{y}\in S_{i,j}^{-}}
        \delta(\bm{y})\cdot \left(f_{\sumpostdata(\noisydata)}(\bm{y})-f_{\sumpostdata(\noisydata)}(\swap{i}{j}{\bm{y}})\right)d\bm{y}\nonumber\,,
    \end{align}
    where Equation \eqref{eq:bias_ineq_long_1} comes from the variable substitution 
    $\bm{y}'=\swap{i}{j}{\bm{y}}$ and algebraic representation of 
    the swapping operation $\swap{i}{j}{\bm{y}}=\bm{P}_{(ij)}\bm{y}$ with the permutation
    matrix $\bm{P}_{(ij)}$. Equation \eqref{eq:bias_ineq_long_2} leverages the result of Lemma
    \ref{lem:swap_invariance} and the fact that the determinant of a permutation matrix has
    absolute value $1$. Recall that Proposition \ref{prop:func_d} shows that
    the function $d$ is decreasing in its first argument and, thus,
    \begin{equation*}
        \delta(\bm{y})=d(y_i,T(\bm{y}))-d(y_j,T(\bm{y}))\geq 0\,,\qquad\forall~\bm{y}\in S_{i,j}^-\,,
    \end{equation*}
    because $y_i$ is no greater than $y_j$ by definition of the set $S_{i,j}^-$.
    Additionally, Corollary \ref{cor:pdf_diff} and Lemma \ref{lem:lap_pdf_diff} show
    that the following inequality holds for any $\bm{y}\in S_{i,j}^-$ under
    the Gaussian and Laplace mechanism respectively
    \begin{equation*}
        f_{\sumpostdata(\noisydata)}(\bm{y})-f_{\sumpostdata(\noisydata)}(\swap{i}{j}{\bm{y}})\geq 0\,.
    \end{equation*}
    These two results above make the joint efforts to establish the following.
    \begin{multline*}
        \EE{\sumpostdata(\noisydata)}{\delta(\sumpostdata(\noisydata))}=\\
        \int_{\bm{y}\in S_{i,j}^{-}}
        \delta(\bm{y})\cdot \left(f_{\sumpostdata(\noisydata)}(\bm{y})-f_{\sumpostdata(\noisydata)}(\swap{i}{j}{\bm{y}})\right)d\bm{y}\geq 0\,.
    \end{multline*}
    \noindent
    By Proposition \ref{prop:func_d}, the function $d$ satisfies the following properties:
    for any $\bm{y}\in \sumregion$,
    \begin{multline*}
        \begin{cases}
            d(y_i,T(\bm{y}))\leq 0\,, \\
            d(y_j,T(\bm{y}))\geq -T(\bm{y})\,,
        \end{cases}\implies\\
        \delta(\bm{y})=d(y_i,T(\bm{y}))-d(y_j,T(\bm{y}))\leq T(\bm{y})\,,
    \end{multline*}
    which implies that 
    \begin{equation*}
        \EE{\sumpostdata(\noisydata)}{\delta(\sumpostdata(\noisydata))}\leq 
        \EE{\sumpostdata(\noisydata)}{T(\sumpostdata(\noisydata))}\,.
    \end{equation*}
    The expectation $\EE{\sumpostdata(\noisydata)}{\delta(\sumpostdata(\noisydata))}$ is thus between $0$ and 
        $\EE{\sumpostdata(\noisydata)}{T(\sumpostdata(\noisydata))}$. The
        equivalence established in Equation \eqref{eq:bias_ineq_equiv} then concludes the proof here.
\end{proof}

%%%%%%%%%%%%%%%%%%%%%%%%%%%%%%%%%%%%%%%%%%%%%%%%%%%%%%%%%%%%%%%%%%%%%%%%%%%%%%%%%%
% Proof of the second upper bound
%%%%%%%%%%%%%%%%%%%%%%%%%%%%%%%%%%%%%%%%%%%%%%%%%%%%%%%%%%%%%%%%%%%%%%%%%%%%%%%%%%

\begin{proof}[Proof of Lemma \ref{thm:bias_ineq_upper}]
    Likewise, by Lemma \ref{lem:kkt_ps}, the right hand side of Equation \eqref{eq:bias_ineq_upper} 
    can be written as 
    \begin{align*}
         &\biasm{\nnsumpostdata}_i-\biasm{\nnsumpostdata}_j\\
        =~&\left(\EE{\sumpostdata(\noisydata)}{\relu{\sumpostdata(\noisydata)_i -
        T(\sumpostdata(\noisydata))}}-x_i\right)
        -\\
        &\left(\EE{\sumpostdata(\noisydata)}{\relu{\sumpostdata(\noisydata)_j - T(\sumpostdata(\noisydata))}}
        -x_j\right)\,.
    \end{align*}
    Let $\rho:\sumregion\mapsto\RR$ denote the function
    \begin{equation*}
        \rho(\bm{y})\coloneqq\relu{y_i - T(\bm{y})}-\relu{y_j - T(\bm{y})}
    \end{equation*}
    such that
    \begin{align*}
        &\left(\biasm{\nnsumpostdata}_i-\biasm{\nnsumpostdata}_j\right)-\left(
        x_j-x_i\right)\\
        =~&\left(\EE{\sumpostdata(\noisydata)}{\relu{\sumpostdata(\noisydata)_i - T(\sumpostdata(\noisydata))}}-x_i\right)
        -\\
        &\left(\EE{\sumpostdata(\noisydata)}{\relu{\sumpostdata(\noisydata)_j - T(\sumpostdata(\noisydata))}}
        -x_j\right)
        -\left(
        x_j-x_i\right)\\
        =~&\EE{\sumpostdata(\noisydata)}{\rho(\sumpostdata(\noisydata))}\,.
    \end{align*}
    Note that the feasible region $\sumregion$ can be partitioned into two
    almost disjoint sets $S_{i,j}^{-}$ and $S_{i,j}^{+}$
    \begin{align*}
        S_{i,j}^{-}&=\sumregion\bigcap H_{i,j}^-=\sumregion\bigcap\left\{\bm{y}\mid
        y_i-y_j\leq 0\right\}\,,\\
        S_{i,j}^{+}&=\sumregion\bigcap H_{i,j}^+=\sumregion\bigcap\left\{\bm{y}\mid
        y_i-y_j\geq 0\right\}\,,
    \end{align*}
    and $\swap{i}{j}{\cdot}$ is a bijection between these two sets. It follows that
    \begin{align}
        &\EE{\sumpostdata(\noisydata)}{\rho(\sumpostdata(\noisydata))}
        \nonumber\\
        =~&\int_{\bm{y}\in S_{i,j}^{-}}
        \rho(\bm{y})\cdot f_{\sumpostdata(\noisydata)}(\bm{y})d\bm{y}+\int_{\bm{y}'\in S_{i,j}^{+}}
        \rho(\bm{y}')\cdot f_{\sumpostdata(\noisydata)}(\bm{y}')d\bm{y}'\nonumber\\
        =~&\int_{\bm{y}\in S_{i,j}^{-}}
        \rho(\bm{y})\cdot f_{\sumpostdata(\noisydata)}(\bm{y})d\bm{y}+\nonumber\\
        &\int_{\bm{y}\in S_{i,j}^{-}}
        \rho(\swap{i}{j}{\bm{y}})\cdot
        f_{\sumpostdata(\noisydata)}(\swap{i}{j}{\bm{y}})\vert\det(\bm{P}_{(ij)})\vert d{\bm{y}}
        \label{eq:bias_ineq_upper_1}\\
        =~&\int_{\bm{y}\in S_{i,j}^{-}}
        \rho(\bm{y})\cdot f_{\sumpostdata(\noisydata)}(\bm{y})+
        \rho(\swap{i}{j}{\bm{y}})\cdot\nonumber\\
        &\qquad \quad f_{\sumpostdata(\noisydata)}(\swap{i}{j}{\bm{y}}) d{\bm{y}}\label{eq:bias_ineq_upper_2}\\
        =~&\int_{\bm{y}\in S_{i,j}^{-}}
        \rho(\bm{y})\cdot \left(f_{\sumpostdata(\noisydata)}(\bm{y})-
        f_{\sumpostdata(\noisydata)}(\swap{i}{j}{\bm{y}})\right)d\bm{y}\label{eq:bias_ineq_upper_3}\\
        \leq~&0\label{eq:bias_ineq_upper_4}\,,
    \end{align}
    where Equation \eqref{eq:bias_ineq_upper_1} comes from the variable substitution 
    $\bm{y}'=\swap{i}{j}{\bm{y}}$ and algebraic representation of 
    the swapping operation $\swap{i}{j}{\bm{y}}=\bm{P}_{(ij)}\bm{y}$ with the permutation
    matrix $\bm{P}_{(ij)}$. Equation \eqref{eq:bias_ineq_upper_2} leverages the result of Lemma
    \ref{lem:swap_invariance} and the fact that the determinant of a permutation matrix has
    absolute value $1$. 
    Notice that
    \begin{align*}
        &\rho(\swap{i}{j}{\bm{y}}) \\
        =~& \relu{y_j - T(\swap{i}{j}{\bm{y}})}-
        \relu{y_j - T(\swap{i}{j}{\bm{y}})}\\
        =~&-\left[\relu{y_i - T(\bm{y})}-\relu{y_j - T(\bm{y})}\right]\\
        =~&-\rho(\bm{y})\,,
    \end{align*}
    where the second equation comes from Lemma \ref{lem:swap_invariance}. Thus, 
    Equation \eqref{eq:bias_ineq_upper_3} can be derived from the formula above.
    The value $\rho(\bm{y})$ turns out to be non-positive for any $\bm{y}\in S_{i,j}^{-}$, i.e.,
    \begin{align*}
        \rho(\bm{y})&=\relu{y_i - T(\bm{y})}-\relu{y_j - T(\bm{y})}\\
        &=
        \relu{y_i - T(\bm{y})}-\relu{\left(y_i - T(\bm{y})\right)+\left(y_j-y_i\right)}\\
        &\leq 0\,,
    \end{align*}
    because of monotonicity of the function $\relu{\cdot}:\RR\mapsto\RR_+$.
    Additionally, Corollary \ref{cor:pdf_diff} and Lemma \ref{lem:lap_pdf_diff} show
    that the following inequality holds for any $\bm{y}\in S_{i,j}^-$ under
    the Gaussian and Laplace mechanism respectively
    \begin{equation*}
        f_{\sumpostdata(\noisydata)}(\bm{y})-f_{\sumpostdata(\noisydata)}(\swap{i}{j}{\bm{y}})\geq 0\,,
    \end{equation*}
    which indicates that Equation \eqref{eq:bias_ineq_upper_4} holds.
\end{proof}

\subsubsection*{Proof of Proposition \ref{thm:lower_explicit}}
\begin{proof}[Proof of Proposition \ref{thm:lower_explicit}]
    Note that, if the injected noise follows a multivariate Gaussian distribution 
    $\cN(\boldsymbol{0}, \sigma^2\bm{I}_n)$, for $i\in [n]$,
    \begin{equation*}
        \nu^* = \eta_i-\frac{\sum_{j=1}^n\eta_j}{n} \sim \cN\left(0, \frac{(n-1)\sigma^2}{n}\right) = \cN\left(0,
        a^{-2}\right)\,,
    \end{equation*}
    where $a = \frac{1}{\sigma}\sqrt{\frac{n}{n-1}}$. Denote the 
    cumulative distribution function and probability density function
    of $\nu^*$ as $F_{\nu^*}(\cdot)$ and $f_{\nu^*}(\cdot)$.
    The bias term 
    $\biasm{\relu{\sumpostdata}}_i$ can then be written as
    \begin{align*}
        &\bias{\relu{\sumpostdata}}{\truedata}_i \\
        =~& \mathbb{E}_{\nu^*}\left[\left(x_i+\nu^*\right)_+\right]-x_i =\int_{-x_i}^{+\infty} (x_i+t) f_{\nu^*}(t) dt - x_i \\
        =~&-\int_{-\infty}^{-x_i} (x_i+t) f_{\nu^*}(t) dt =- \int_{-\infty}^{-x_i} (x_i+t) d F_{\nu^*}(t) \\
        =~&-(x_i+t)F_{\nu^*}(t)\bigg|_{-\infty}^{-x_i}+\int_{-\infty}^{-x_i}F_{\nu^*}(t) dt.
        % &=x_1 
    \end{align*}
    Notice that
    \begin{align*}
        \lim_{t\rightarrow -\infty} (x_i+t)F_{\nu^*}(t) &= \lim_{t\rightarrow -\infty} \frac{F_{\nu^*}(t)}{t^{-1}} =
        -\lim_{t\rightarrow -\infty} \frac{f_{\nu^*}(t)}{t^{-2}} \\
        &= -\lim_{t\rightarrow -\infty}
        \frac{t^2}{\sqrt{2\pi}a^{-1}\exp(a^2t^2/2)} \\
        &= 0\,.
    \end{align*}
    Therefore, 
    \begin{align*}
        &\bias{\relu{\sumpostdata}}{\truedata}_1 - \bias{\relu{\sumpostdata}}{\truedata}_n \\
        =~& \int_{-\infty}^{-x_1}F_{\nu^*}(t) dt - \int_{-\infty}^{-x_n}F_{\nu^*}(t) dt
        =\int_{-x_n}^{-x_1}F_{\nu^*}(t) dt \\
        =~& \int_{-x_n}^{-x_1}\Phi(at) dt\,,
    \end{align*}
    where $\Phi(\cdot)$ is the standard Gaussian cumulative distribution function. 
    According to mean value theorem, there exists $x_c\in[-x_n,x_1]$, such that
    \begin{align*}
        \int_{-x_n}^{-x_1}\Phi(at) dt&=\Phi(ax_c) (x_n-x_1) \geq \Phi(-ax_n) (x_n-x_1)\,,\\
        \int_{-x_n}^{-x_1}\Phi(at) dt&=\Phi(ax_c) (x_n-x_1) \leq \Phi(-ax_1) (x_n-x_1)\,,
    \end{align*}
    where the inequalities use the fact that $\Phi(\cdot)$ is an increasing function.
\end{proof}

%%%%%%%%%%%%%%%%%%%%%%%%%%%%%%%%%%%%%%%%%%%%%%%%%%%%%%%%%%%%%%%%%%%%%%%%%%%%%%%%%%%%%%%%%%%%%%
\subsubsection*{Proof of Theorem \ref{thm:a-fair}}
\label{app:proof_of_sec_alloc}
%%%%%%%%%%%%%%%%%%%%%%%%%%%%%%%%%%%%%%%%%%%%%%%%%%%%%%%%%%%%%%%%%%%%%%%%%%%%%%%%%%%%%%%%%%%%%%
 \begin{proof}[Proof of Theorem \ref{thm:a-fair}]
 
    Let $\tilde{f}$ denote the objective function of program \eqref{prog:s}, i.e.,
    $\tilde{f}(\bm{v}) = \norm{\bm{v} - \noisyalloc}_{\rightleftharpoons}$.
    Then, the proof here is split into two parts.
   \begin{itemize}
       \item $\noisyalloc\geq \bm{0}$;\\
       In this case, the mechanism $\posmech{\noisydata}$ simply produces $\noisyalloc$ as output
       and thus
       \begin{equation*}
           \tilde{f}\left(\posmech{\noisydata}\right) = \norm{\noisyalloc - \noisyalloc}_{\rightleftharpoons}=0\,.
       \end{equation*}
       Therefore, due to non-negative of the objective function $\tilde{f}$,
       the mechanism $\posmech{\noisydata}$ generates an optimal solution to program \eqref{prog:s}.
       \item $\exists~i\in[n],~\nalloci{i}<0$;\\
       Suppose that there exists $\bm{u}\in\Delta_n$ such that
       $\tilde{f}(\bm{u})<\tilde{f}(\posmech{\noisydata})$.
       In order to complete the proof here, it suffices to
       demonstrate that the very 
       existence of such $\bm{u}$ is invalid. In
       the rest of this proof, most of efforts will be put
       to show that this vector $\bm{u}$ must be summed up to
       a positive constant strictly larger than $1$, which means that
       it is unlikely for $\bm{u}$ to lie in the probability simplex
       $\Delta_n$.\\
       Let $\cI^-=\{j\mid \nalloci{j}<0\}$ be the index set,
       which captures all the negative entries of $\noisyalloc$. 
       By Lemma \ref{lem:kkt_ps}, since the mechanism $\pos$
       makes use of projection to restore feasibility,
       its output $\posmech{\noisydata}$ can thus be expressed as
       \begin{equation*}
           \posmech{\noisydata} = \relu{\nalloci{j}-T\cdot \bm{1}}\,,\qquad\exists~T\in\RR_+\,,
       \end{equation*}
       where $T$ satisfies the following relation
       \begin{equation*}
           \sum_{j=1}^n\relu{\nalloci{j}-T}=1\,.
       \end{equation*}
       Thus, it follows that
       \begin{equation}\label{eq:pos_mech}
           \begin{aligned}
              \tilde{f}(\posmech{\noisydata})=&\max_{j\in [n]}~\left(\posmech{\noisydata}_j-\nalloci{j}\right)-\\
              &\min_{j\in [n]}~\left(\posmech{\noisydata}_j-\nalloci{j}\right)\\
               &=\max_{j\in\cI^-}~\left(0-\nalloci{j}\right)-(-T)\\
               &=-\min_{j\in\cI^-}~\nalloci{j} + T\,,
           \end{aligned}
       \end{equation}
       and
       \begin{equation}
           \begin{aligned}
               &\sum_{j\in\cI^-}\left(\posmech{\noisydata}_j-\nalloci{j}\right)=-\sum_{j\in\cI^-}\nalloci{j}\\
           =~&\sum_{j\in [n]\setminus \cI^-} \left[\nalloci{j}-\relu{\nalloci{j}-T}\right]\\
           =~&-\sum_{j\in [n]\setminus \cI^-} \left(\posmech{\noisydata}_j-\nalloci{j}\right)\,.
           \end{aligned}\label{eq:pos_neg_sum}
       \end{equation}
       Since $\bm{u}$ lies in the probability simplex $\Delta_n$, 
       $\bm{u}$ is a non-negative vector, which implies that
       \begin{multline}\label{eq:max_ineq}
           \max_{j\in[n]}~\left(u_j-\nalloci{j}\right)\geq \max_{j\in\cI^-}~\left(u_j-\nalloci{j}\right)\geq\\
           \max_{j\in\cI^-}~\left(0-\nalloci{j}\right)=-\min_{j\in\cI^-}~\nalloci{j}\,.
       \end{multline}
       As a consequence, 
       \begin{align}
           &\min_{j\in[n]}~\left(u_j-\nalloci{j}\right)\nonumber\\
           =~&\max_{j\in[n]}~\left(u_j-\nalloci{j}\right)-\tilde{f}(\bm{u})\nonumber\\
           \geq~& -\min_{j\in\cI^-}~\nalloci{j} - \tilde{f}(\bm{u})\label{eq:sec5-1}\\
           =~&\left(\tilde{f}(\posmech{\noisydata})-T\right)-\tilde{f}(\bm{u})\label{eq:sec5-2}\\
           =~&\left(\tilde{f}(\posmech{\noisydata})-\tilde{f}(\bm{u})\right)-T\nonumber\\
           >~&-T\,,\label{eq:sec5-3}
       \end{align}
       where Equation \eqref{eq:sec5-1} comes from Equation \eqref{eq:max_ineq} and 
       Equation \eqref{eq:sec5-2} is a direct consequence of Equation \eqref{eq:pos_mech}. Moreover,
       Equation \eqref{eq:sec5-3} comes from the assumption that $\tilde{f}(\bm{u})<\tilde{f}(\posmech{\noisydata})$.
       Let $t$ be a shorthand for $-\min_{j\in[n]}~\left(u_j-\nalloci{j}\right)$, which leads to the following
       argument that, for any $j\in[n]$,
       \begin{multline*}
           \begin{cases}
                  u_j>\nalloci{j}-t\,,\\
                  u_j\geq 0
           \end{cases}\implies\\
           u_j\geq \relu{\nalloci{j}-t}\geq \relu{\nalloci{j}-T}\,.
       \end{multline*}
       Besides, there exists at least an index $l\in [n]\setminus\cI^-$ such that
       \begin{multline*}
           \relu{\nalloci{l}-t}=\nalloci{l}-t>\nalloci{l}-T=\\\relu{\nalloci{l}-T}>0=\posmech{\noisydata}_l\,,
       \end{multline*}
       because at least one entry of $\posmech{\noisydata}\in\Delta_n$ is positive. It follows that
       \begin{equation}\label{eq:strict_ineq}
           \sum_{j\in[n]\setminus\cI^-} u_j\geq  \sum_{j\in[n]\setminus\cI^-} \relu{\nalloci{j}-t}>
           \relu{\nalloci{j}-T}\,.
       \end{equation}
       With the analysis above, it is finally a step away from a contradiction.
       \begin{align}
           &\sum_{j\in [n]} \left(u_j-\nalloci{j}\right)\nonumber\\
           =~&\sum_{j\in\cI^-}\left(u_j-\nalloci{j}\right)+
           \sum_{j\in [n]\setminus \cI^-}\left(u_j-\nalloci{j}\right)\nonumber\\
           \geq~& -\sum_{j\in\cI^-}\nalloci{j}+
           \sum_{j\in [n]\setminus \cI^-}\left(u_j-\nalloci{j}\right)\label{eq:sec5-4}\\
           >~&-\sum_{j\in\cI^-}\nalloci{j}+\nonumber\\
          & \sum_{j\in [n]\setminus \cI^-}\left(\relu{\nalloci{l}-T}-\nalloci{j}\right)\label{eq:sec5-5}\\
           =~&-\sum_{j\in\cI^-}\nalloci{j}+
           \sum_{j\in [n]\setminus \cI^-}\left(\posmech{\noisydata}_j-\nalloci{j}\right)\nonumber\\
           =~&-\sum_{j\in\cI^-}\nalloci{j}+\sum_{j\in\cI^-}\nalloci{j}\label{eq:sec5-6}\\
           =~&0\,,\nonumber
       \end{align}
       where Equation \eqref{eq:sec5-4} comes from non-negativity of $\bm{u}$ and Equation \eqref{eq:sec5-5}
       is a direct consequence of Equation \eqref{eq:strict_ineq}. Additionally, 
       Equation \eqref{eq:sec5-6} leverages the result
       of Equation \eqref{eq:pos_neg_sum}. Therefore, it leads to the following
       \begin{equation*}
           \sum_{j\in[n]} u_j >\sum_{j\in[n]} \nalloci{j} = 1\,,
       \end{equation*}
       which is contradictory to the assumption that the vector $\bm{u}$ is inside the probability simplex $\Delta_n$. It concludes the proof here.
   \end{itemize}
  \end{proof}

\subsubsection*{Proof of Proposition \ref{prop:cop_as_l1_norm}}  
\begin{proof}[Proof of Proposition \ref{prop:cop_as_l1_norm}]
    Since $\pp\in\Pi_{\Delta_n}$, it follows that, for any noisy data $\noisydata\in\RR^n$,
    \begin{equation*}
        \sum_{i=1}^n \pp(\noisydata)_i = \sum_{i=1}^n \talloci{i} = 1\,,
    \end{equation*}
    and thus
    \begin{multline*}
        \sum_{i=1}^n \bias{\pp}{\al}_i = 0\implies \\\sum_{i\in [n]\setminus \cJ^-}
        \bias{\pp}{\al}_i = -\sum_{j\in\cJ^-} \bias{\pp}{\al}_j\,. 
    \end{multline*}
    Therefore, by the definition of $\ell_1$ norm,
    \begin{align*}
        &\sum_{j\in \cJ^-} \vert \bias{\pp}{\al}_j\vert \cdot B =
        -B\sum_{j\in \cJ^-}\bias{\pp}{\al}_j\\
        =~&\frac{B}{2}\cdot
        \left(- \sum_{j\in \cJ^-}\bias{\pp}{\al}_j+\sum_{i\in [n]\setminus \cJ^-}
        \bias{\pp}{\al}_i \right)\\
        =~&\frac{B}{2}\cdot\norm{\bias{\pp}{\al}}_1\,.
    \end{align*}
\end{proof}

\subsection*{Proof of Theorem \ref{thm:cop}}
\begin{proof}[Proof of Theorem \ref{thm:cop}]
    Let $\tilde{h}$ be a shorthand for the objective function of program \eqref{prog:s-cop}, i.e.,
    $\tilde{h}(\bm{v}) =B/2\cdot \norm{\bm{v}-\noisyalloc}_1$.
    For simplicity, this work only presents the detailed proof for the projection onto
    simplex mechanism. The result of the baseline mechanism can be obtained with ease by using
    similar techniques. \\
    Suppose that $\noisyalloc$ is a non-negative vector. The post-processed count
    $\posmech{\noisydata}$ is simply identical to $\noisyalloc$ and thus becomes the optimal solution
    to program \eqref{prog:s-cop} because $B/2\cdot\norm{\posmech{\noisydata}-\noisyalloc}_1=0$.\\
    Consider the non-trivial scenario where there exists at least one index $k\in[n]$
    such that $\nalloci{k}<0$. Let $\cI^+=\{j\mid \nalloci{j}>0\}$ denote the index set, which captures 
    all of the positive entries of $\noisyalloc$. Then, it follows that, for any
    $\bm{v}\in \Delta_n$,
    \begin{align}
        \tilde{h}(\bm{v})&=\frac{B}{2}\cdot \norm{\bm{v}-\noisyalloc}_1=\frac{B}{2}\sum_{i=1}^n\left\vert
        v_i-\nalloci{i}\right\vert\nonumber\\
        &=\frac{B}{2}\sum_{i\in \cI^+}\left\vert
        v_i-\nalloci{i}\right\vert+\frac{B}{2}\sum_{i\in [n]\setminus\cI^+}\left\vert
        v_i-\nalloci{i}\right\vert\nonumber\\
        &=\frac{B}{2}\sum_{i\in \cI^+}\left\vert
        v_i-\nalloci{i}\right\vert+
        \frac{B}{2}\sum_{i\in [n]\setminus\cI^+}
        \left(v_i-\nalloci{i}\right)\nonumber\\
        &\geq \frac{B}{2}\sum_{i\in \cI^+}
        \left(\nalloci{i}-v_i\right)+
        \frac{B}{2}\sum_{i\in [n]\setminus\cI^+}
        \left(v_i-\nalloci{i}\right)\label{eq:cop_long_1}\\
        &=\frac{B}{2}\left(\norm{\noisyalloc}_1+1\right)-B\sum_{i\in\cI^+}v_i\nonumber\\
        &\geq \frac{B}{2}\left(\norm{\noisyalloc}_1-1\right)\,,\nonumber
    \end{align}
    where Equation \eqref{eq:cop_long_1} comes from the triangle inequality.
    By Lemma \ref{lem:kkt_ps}, the post-processed count $\posmech{\noisydata}$ is known
    to have the following form
    \begin{equation*}
        \posmech{\noisydata}=\relu{\noisyalloc-T\cdot \bm{1}}\,,\qquad\exists~T\in \RR_+\,,
    \end{equation*}
    such that
    \begin{equation*}
        \sum_{i=1}^n \posmech{\noisydata}_i=\sum_{i=1}^n \relu{\nalloci{i}-T}=1\,,
    \end{equation*}
    and, therefore,
    \begin{equation*}
        \posmech{\noisydata}_i=\begin{cases}
               \relu{\nalloci{i}-T}<\nalloci{i}\,, & i\in\cI^+\,,\\
               0\,, & \text{otherwise.}
        \end{cases}
    \end{equation*}
    The value $\tilde{h}(\posmech{\noisydata})$ can then be computed as
    \begin{align*}
        \tilde{h}(\posmech{\noisydata})&=\frac{B}{2}\left[\!\sum_{i\in \cI^+}\!\!
        \left(\nalloci{i}-\posmech{\noisydata}_i\right)-\!\!\!\!\!\!\sum_{i\in [n]\setminus\cI^+}\!\!\!\!\nalloci{i}\right]\\
        &=\frac{B}{2}\norm{\noisyalloc}_1-\frac{B}{2}\sum_{i\in\cI^+}\posmech{\noisydata}_i\\
        &=\frac{B}{2}\left(\norm{\noisyalloc}_1-1\right)\,,
    \end{align*}
    which means that $\tilde{h}(\posmech{\noisydata})$ equals the lower bound of the objective $\tilde{h}$
    over the probability simplex $\Delta_n$. Hence, the projection onto simplex mechanism
    produces an optimal solution to program \eqref{prog:s-cop}.
\end{proof}

\end{document}